\documentclass[10pt,journal,compsoc]{IEEEtran}
\usepackage{booktabs} 
\usepackage{amsthm}
\usepackage{tabularx}
\usepackage{amsmath}
\usepackage{comment}
\usepackage{graphicx}
\usepackage{threeparttable}
\usepackage{algorithmicx}
\usepackage[linesnumbered, ruled, vlined]{algorithm2e}
\usepackage[noend]{algpseudocode}
\usepackage{amsmath}
\usepackage{indentfirst}
\usepackage{multirow}
\usepackage{array}
\usepackage{xcolor}

\definecolor{azure(colorwheel)}{rgb}{0.0, 0.5, 1.0}
\definecolor{frenchblue}{rgb}{0.0, 0.45, 0.73}
\definecolor{bittersweet}{rgb}{1.0, 0.44, 0.37}
\definecolor{green(pigment)}{rgb}{0.0, 0.65, 0.31}
\definecolor{navyblue}{rgb}{0.0, 0.0, 0.5}
\definecolor{purple(html/css)}{rgb}{0.5, 0.0, 0.5}

\newcommand{\rmnum}[1]{\romannumeral #1}
\newcommand{\Rmnum}[1]{\expandafter\@slowromancap\romannumeral #1@}
\newcolumntype{C}[1]{>{\centering\arraybackslash}m{#1}}

\newtheorem{definition}{Definition}[section]
\newtheorem{lemma}{Lemma}[section]
\newtheorem{theorem}{Theorem}[section]
\newtheorem{property}{Property}
\newtheorem{assumption}{Assumption}
\newtheorem{problem}{Problem}
\newtheorem{observation}{Observation}[section]

\begin{document}
\title{Performance Analysis of QoS-Differentiated Pricing in Cloud Computing: An Analytical Approach}
\author{Xiaohu~Wu~Francesco~De~Pellegrini
\IEEEcompsocitemizethanks{\IEEEcompsocthanksitem The authors are with Fondazione Bruno Kessler (FBK), 38123 Povo, Trento, Italy.\protect E-mail: xiaohuwu@fbk.eu, fdepellegrini@fbk.eu}
\thanks{Manuscript received xxx, 2018; revised yyy, 2018.}
}
\markboth{JOURNAL OF LATEX CLASS FILES,~Vol.~X, No.~Y, August~2018}%
{Wu \MakeLowercase{\textit{et al.}}: Performance Analysis of QoS-Differentiated Posted Pricing in Cloud Computing}



\IEEEtitleabstractindextext{%
\begin{abstract}
A fundamental goal in the design of IaaS service is to enable both user-friendly and cost-effective service access, while attaining high resource efficiency for revenue maximization. QoS differentiation is an important lens to achieve this design goal. In this paper, we propose the first analytical QoS-differentiated resource management and pricing architecture in the cloud computing context; here, a cloud service provider (CSP) offers a portfolio of SLAs. In order to maximize the CSP's revenue, we address two technical questions: ({\romannumeral 1}) how to set the SLA prices so as to direct users to the SLAs best fitting their needs, and, ({\romannumeral 2}) determining how many servers should be assigned to each SLA, and which users and how many of their jobs are admitted to be served. We propose optimal schemes to {\em jointly} determine SLA-based prices and perform capacity planning in polynomial time. Our pricing model retains high usability at the customer's end. Compared with standard usage-based pricing schemes, numerical results show that the proposed scheme can improve the revenue by up to a five-fold increase.
\end{abstract}


}

\maketitle

\section{Introduction}

\subsection{Motivation}\label{sec.motivation}

The worldwide cloud Infrastructure-as-a-Service (IaaS) market is projected to grow to \$71.6 billion in 2020 from \$25.3 billion in 2016~\cite{gartner-rp-1}, and is attracting all users with different requirements and budgets to run their applications on cloud infrastructures. Here, IaaS is used as a means of delivering value to users by facilitating their access to computing resources with no need to own and maintain a whole infrastructure and users only need to pay for the time duration in which computing infrastructures are consumed. The design of IaaS obeys fundamental requirements in service design~\cite{service-design}: plan and organize people, infrastructure, communication and material components of the service, aiming at excellence along three dimensions, namely, {\em the service quality}, {\em the system's efficiency}, and {\em the interaction between the service provider and its users}. This requires a joint effort to understand both the needs of customers and implement efficient resource management.

\vspace{0.071em}\noindent \textbf{\rmnum{1}) Interaction.} In the context of IaaS, a pricing model defines the interaction process between a Cloud Service Provider (CSP) and its users and the particular conditions under which computing services are offered to users; from a user's viewpoint, it much determines the IaaS usability \cite{usability}. There is a family of theoretical models relying on auction theory, and an example is the spot instances in Amazon EC2, whose prices vary over time unpredictably \cite{Ben-Yehuda}. Each user bids a price for spot instances, but can get them only if the spot price is below the bid price \cite{amazon-pricing}; here, users may lose their spot instances suddenly and are charged according to the uncertain spot prices in the period of use.

But, most users may prefer to have exact knowledge of the actual price at the time of service purchase. For them, to act as price-takers is the best option in order to consume their preferred service: they can immediately understand the price-setting method, participate in with a cost-optimal strategy, and know the paid price in advance. By design, the posted pricing method has such desirable payment properties, which have motivated a few studies~\cite{Zhang17a,Einav16c}. For instance, in the Google cloud, the price is posted to users in advance, and each user pays the price times the amount of usage, referred to as \textit{usage-based pricing} \cite{google-pricing}. So, in this paper we will investigate the cloud resource provisioning problem in the posted pricing context.



\begin{figure}[h]
  \centering
  \includegraphics[width=3.3in]{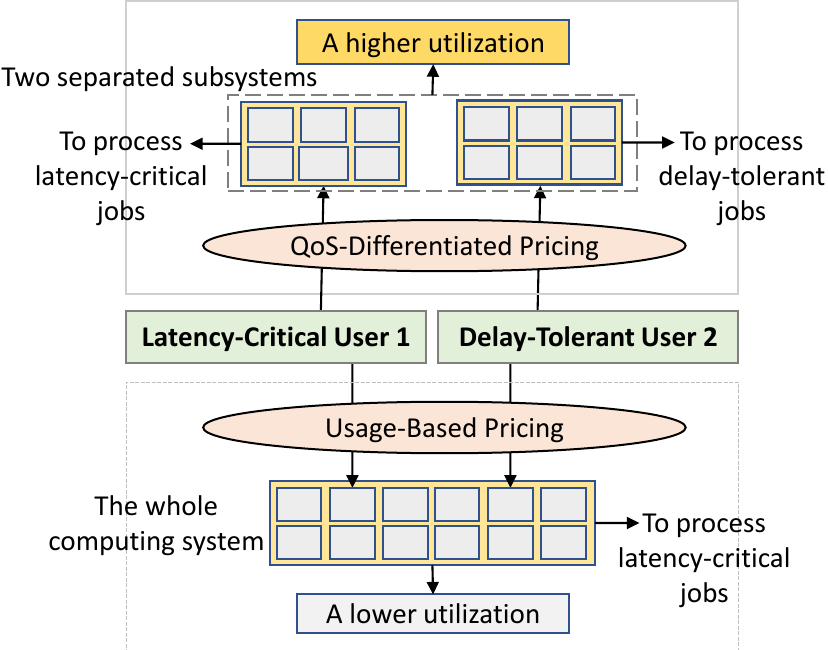}
  \caption{A main motivation of this paper: QoS-differentiated pricing incentivizes user 2 to express its demand of computing resource as delay-tolerant jobs, whereas it is not willing to do so under usage-based pricing.}\label{Fig.4}
\end{figure}

\vspace{0.093em}\noindent \textbf{\rmnum{2}) Resource Efficiency.} Once users' jobs are accepted under a certain pricing mechanism, a scheduling problem appears: it concerns the computing system, and involves dispatching jobs to servers for processing, aiming to maximize resource efficiency. From a queuing perspective, the system's utilization is heavily constrained by the characteristics of jobs being served and could increase as their tolerance to latency is enlarged \cite{Zheng16a,Daigle05,Rasley16a}. Pricing design is an appropriate tool to incentivize users to express their job's characteristics in a way to improve resource efficiency. In the current usage-based pricing solutions, users may prefer to process their jobs immediately upon arrival and not be willing to tolerate latency margins. This happens even when jobs are actually delay-tolerant. In fact, current models do not provide any incentives for elasticity, e.g. discounts for tolerance on task completion delays. In fact, latency-critical and delay-tolerant jobs coexist in the cloud. For instance, cloud jobs from large scale web applications require low response time. Other jobs such as big data analytics or Google crawling data processing have batch processing nature and can tolerate moderate completion delays.


\vspace{0.1em}\noindent{\bf The proposed QoS-differentiated architecture.} To sum up, users of different sources are diverse in terms of their latency requirements when processing jobs. A basic resource management architecture arises in the era of cloud computing where a computing system is shared among users: the whole system is partitioned into several subsystems and each subsystem is used to process jobs with similar latency requirements; in this paper, each subsystem guarantees a finite delay in the completion of jobs and users can choose one of the subsystems that best fits their needs.

For example, one subsystem could be used to process latency-critical jobs, whereas, by offering discounts to users, the other subsystem is to process delay-tolerant jobs, as illustrated in Fig.~\ref{Fig.4}. Let us define $m_1$ ($m_2$) the number of servers in the first (second) subsystem and $r_1$ ($r_2$) the server utilization fraction. Then, the average utilization of the $m_{1}+m_{2}$ servers is $\left( m_{1}\cdot r_{1} + m_{2}\cdot r_{2} \right)/ \left( m_{1} + m_{2} \right)$. However, with no partitioning and no incentive to delay tolerance, utilization only is $r_{1}$ because all users prefer to complete their jobs immediately. With a more intuitive example: if $m_{1}=m_{2}$, $r_{1}=0.1$, and $r_{2}=0.5$, the partitioned system has a utilization up to $0.3$, three times the system's utilization without partitioning.

Several questions come with this architecture. What a CSP concerns most is often the revenue, rather than the utilization alone. However, lower (unit) prices are associated with delay-tolerant jobs in return for the increased utilization; so, does the increased utilization achieved by such a architecture must mean an increase in the CSP's revenue, in contrast to a purely usage-based pricing? Here, the revenue obtained from every server in a unit of time is the product of its utilization and the associated price, where the utilization represents the amount of workload processed in a unit of time. Answering the above question involves optimizing two basic aspects of this architecture: (\rmnum{1}) how much discount should be offered in order for users to submit their jobs as delay-tolerant ones, and (\rmnum{2}) how many servers should be assigned to each subsystem; the latter determines the total amount of jobs it can process under some latency requirement. In this paper, we consider the case where a CSP holds a fixed number of homogeneous servers.

Last, the heterogeneity of users in latency requirements and the potential power of such heterogeneity in resource management entails the establishment of a fundamental QoS-differentiated resource management and pricing architecture. This can be better perceived by conversely understanding why it does not arise in traditional computer systems. Traditionally, they are privately possessed by a user and used for some special purpose, e.g., Google crawling data processing; the processed jobs are possibly highly similar in terms of latency requirements. Even if a user needs to process jobs with different latency requirements, it does not need external incentives and will automatically utilize the potential delay-tolerant nature to improve resource utilization; here, a user usually manages to purchase and possess enough servers to satisfy all its computing need. However, in cloud computing, the CSP is a seller who has a shared system to provide computing services; the pricing aspect is needed to price the shared resource differently among users, providing necessary incentives for users to express their diversity in latency requirements, and the CSP also needs to determine which subset of users are worth serving, subject to the capacity constraint.



\subsection{Main Results}

The main aim of this paper is to demonstrate the potential of such a QoS-differentiated architecture via a standard analytical approach to the analysis of its performance \cite{LeBoudec10a}.
The contributions of this paper are summarized as follows.

\vspace{0.2em}\noindent\textbf{(\rmnum{1})} \textbf{Model Overview} (Section~\ref{sec.model}).\enskip Formally, in the proposed pricing model, a CSP posts to users a menu of $L$ SLAs. The $l$-th SLA consists of a response time $\varphi_{l}$ and its price $\theta_{l}$. The larger the response time, the lower the price. If user $i$ chooses the $l$-th SLA, the CSP guarantees that, upon arrival of a job $j$ of his, such job has maximum waiting time $\varphi_{l}$ and then begins being processed until the completion. The price of processing a unit of workload under SLA $l$ is denoted $\theta_{l}$. Here, the model is user-friendly and of simplicity from a user's perspective: it simply selects a SLA that best fits its need, maximizing its utility minus the price, and its cost (the price times the amount of processed workload) is predictable; it may also be more cost-effective since some users now have opportunities to take advantage of their delay-tolerant nature in return for a degraded price.




The pricing model is an interface between the CSP and its users, determining the acceptance of arriving jobs, and is further linked to an internal resource management scheme to process jobs. We partition the whole computing system into multiple subsystems, each of them serving a specific SLA and its jobs. At each subsystem, a load balancing policy is needed to evenly dispatch the accepted jobs to different servers for processing.

The decision on which policy to use itself is an active research topic in cloud computing and in this paper our focus is on the necessity of establishing a QoS-differentiated architecture and the application of an analytical approach to analyze its performance.
What we do is enabling our analysis applicable to a class of policies that are well applied in cloud environments or mentioned to be promising: they are such that, obeying a SLA, the maximum rate of accepting jobs at a subsystem is linearly proportional to the number of assigned servers; it includes policies such as {\em Random}, {\em Power of Two Choice} (PTC), and {\em Join-the-Idle-Queue} (JIQ) \cite{Jennings15a}. For example, the {\em Random} policy has ever been applied to analyze the viability of cloud brokerage services \cite{Zheng16a} and to a famous simulator CloudAnalyst for balancing workloads among virtual machines \cite{Wickremasinghe10a}; so does the PTC policy in the practice of large-scale cluster management \cite{Ousterhout13a}.

\vspace{0.2em}\noindent\textbf{(\rmnum{2})} \textbf{Analytical Results} (Sections~\ref{sec.system-optimizations} and~\ref{sec.optimization}).\enskip In this paper, our objective is maximizing the CSP's revenue. To do so, we propose an optimal pricing scheme to determine the prices of SLAs; SLA prices -- each SLA corresponding to a specific subsystem -- affect the choices of users and thus determine the job arrival rate at each subsystem. Further, we propose an optimal capacity planning scheme to determine the rate of accepting jobs at every subsystem, in connection with the number of assigned servers. In the following, we introduce the high-level ideas in the process of deriving these two schemes and what makes our derivation non-trivial.



In our analysis, we first consider the case where the SLA prices are predetermined and solve the second problem. We formulate the original problem as an integer linear program (ILP) and analyze the structure of an optimal solution to the ILP. We further reduce the ILP problem to a series of simpler problems whose optimal solution is observable finally. As a result, we provide an optimal capacity planning scheme for the second problem. Next, we analyze the structural properties related to optimal SLA prices. Here, the challenge arises when there exist multiple SLAs. In contrast, when the whole system only offers a single SLA, let us consider the users who succeed in getting services from the CSP and the SLA's optimal price is simply a value close to but smaller than the lowest among these served users' utilities under this SLA; however, this is not the case of this paper where users will make comparisons among all SLAs and select one SLA that maximizes its surplus, i.e., its utility minus the paid price. Relying on the optimal capacity planning scheme, we propose a procedure to obtain an optimal solution to the first problem. 


\vspace{0.2em}\noindent\textbf{(\rmnum{3})} \textbf{Performance Evaluation} (Section~\ref{sec.performance-evaluation}).\enskip 
The proposed architecture is compared with the usage-based pricing model; in fact, the latter corresponds to a special case of the proposed QoS-differentiated pricing where $L=1$, i.e., only a SLA is provided, and our model is complementary to that model by QoS differentiation. Simulations show that the revenue of a CSP could be improved by up to a five-fold increase, ranging from 499.2\% to 44.09\%. In addition, by numerical results, we also analyze how the performance of a QoS-differentiated posted pricing model is affected by the number of servers that a CSP possesses, and the market environments such as the total amount of the users' demands of computing services, the user population's sensitivity to latency (how fast their utility/willingness-to-pay decreases with the increment of delay), the diversity of users (how much their values vary) etc.

Finally, although in this paper we analyze and discuss the proposed QoS-differentiated pricing and resource management architecture in the case where a CSP holds a fixed number of servers, it could also be used to analyze under a particular market environment what is the optimal number of servers that a CSP should possess and provision to users, in terms of the maximum revenue it can obtain from the whole market or other metrics. This can be achieved by using numerical results to observe how its maximum revenue changes with the number of servers that it possesses.


\section{Related Works}\label{sec.related-works}

In this section, we review existing works on pricing of IaaS services.

\vspace{0.25em}\noindent\textbf{Analytical Modeling Spot Pricing.} Investigating the potential models for providing and pricing cloud services and their effectiveness are very real and important questions for the providers \cite{Devanur17}. Methodologically, of the most relevance to our work is the existing research line which investigates which model is more effective: offering on-demand instances alone or offering both spot and on-demand instances to users.
As stated above, spot instances are priced by auctions while on-demand instances are charged a fixed price per unit of time when utilizing a virtual machine, i.e., usage-based pricing.

Such works also based an analytical approach to performance analysis on queuing theory, the standard tool for evaluating complex computing and communication systems \cite{LeBoudec10a}.
In particular, Abhishek {\em et al.} modeled the hybrid market of Amazon as a queuing system, further analyzed by auction theory \cite{Abhishek12}. They use a linear function of a job's waiting time to characterize the job's utility: under the assumption that the CSP has access to an infinite pool of resources, they have showed that the introduction of spot instances could not increase the revenue of a CSP. The reason for this is that, in a hybrid market, there is no way to prevent high-value low-waiting-time-cost jobs from choosing the spot market when they would have been willing to pay a higher on-demand price. More recently, L. Dierks, and S. Seuken considered a more realistic setting where the amount of instances is finite. Under such assumption, they have showed that a hybrid market can indeed lead to higher revenue than a single on-demand market \cite{Dierks16a}. 

Beyond spot pricing, in the recent literature, analytical models are also proposed to study the viability of cloud brokerage services under usage-based pricing~\cite{Zheng16a}. In our work, similar job queuing and user's utility models are used. However, we consider the case where multiple service level agreements (SLAs) are offered to users whereas in \cite{Zheng16a} a single SLA is assumed. As a consequence, in our framework, new combined problems of pricing and capacity planning arise. In addition, the QoS-differentiated posted pricing is tested under three analytical policies, whereas in \cite{Zheng16a} only the Random policy is considered.

\vspace{0.2em}\noindent\textbf{Preliminary Measurement of QoS-Differentiated Pricing.}
Recently, the effectiveness of QoS-differentiated posted pricing has been validated via an experimental deployment \cite{Sandholm15}. The CSP, based on the current load, can specify two completion times per submitted job. The corresponding prices are determined accordingly, where the later a job is completed, the cheaper the price. The experimental results show that such a pricing system can yield 40\% improvement over the standard fixed pricing scheme. However, the pricing scheme in \cite{Sandholm15} is cloud-centric and not designed to satisfy per user QoS requirements. In other words, when the available servers are not enough, the CSP is forced to announce large job completion times to all users. This happens even though some of their jobs may be latency-critical. The latter observation suggests instead to consider a pre-defined set of QoS requirements, including a wide range of completion times. Then, a user can choose a QoS requirement that fits best its need, and the CSP, in turn, has to ensure that the user's jobs are completed satisfying the intended QoS requirements.

Furthermore, in terms of methodologies, performance evaluation proposed in \cite{Sandholm15} is based on measurements. Conversely, we based our evaluation on analytical modeling. Obeying to standard performance evaluation methodologies \cite{LeBoudec10a}, we complement our analysis with numerical experiment in order to show the viability of QoS-differentiated pricing.

\vspace{0.25em}\noindent\textbf{Mechanism Design.}
Mechanism design methods have attracted substantial attention so far\footnote{Although most of these mechanisms consider the objective of maximizing social welfare, it could be related to the revenue maximization objective of this paper through some existing techniques \cite{Cai13a}.}. In this type of schemes, each user reports to the cloud its job requirements, e.g., value, workload, and deadline by which a job has to be completed; a CSP determines which users to be served and the prices paid by users. One line of works considers either a static scheme, i.e., by which all servers are sold at one round, or the case where each user specifies a rigid time interval and the number of servers requested per time slot \cite{zhang14c,zhang14b,shi14a,Zhang15}. %
Another line of works is further motivated by the fact that the jobs in cloud computing have different time elasticity due to diverse computing needs. Such works typically advocate the case for elastic scheduling in presence of delay-tolerant jobs~\cite{Jain15a,Lucier,Azar15,Zhou16}. In this case, each user specifies its job's workload and a large time interval; however, the user only cares about completing its job in the specified interval. This grants to the CSP additional flexibility in order to decide in which slots the job will be executed; however, all these jobs are still processed in the same (whole) system and are not differentiated by their deadline requirements.

It is worth noting that these mechanisms are typically evaluated using worst-case analysis. The standard method is to evaluate the performance ratio between the proposed mechanisms and an ideal, optimal one, under the same pricing and resource allocation model. In our work, conversely, we focus on which model is better in terms of its usability and the generated revenue and particularly on the performance improvement obtained when QoS-differentiated posted pricing is compared with the standard usage-based pricing method, which does not offer incentives to delay their job's completion. By theoretical results and numerical outcomes we justify that such pricing scheme is indeed able to attain higher resource efficiency. Finally, the ideally optimal configuration of such pricing is provided, based on a queuing model.

\vspace{0.25em}\noindent\textbf{Posted Pricing.} Posted pricing has been introduced in Section~\ref{sec.motivation}. Such scheme is customer-friendly and removes the complexity in auction-based pricing (see also Google pricing \cite{google-pricing} for further details). Recently, Li {\em et al.} advanced the theoretical foundation of posted pricing in cloud computing~\cite{Zhang17a}. The authors have designed a type of pricing function based on the resource utilization ratio, and used a worst-case analysis to derive the competitive ratio, i.e., the ratio of the social welfare achieved by the pricing function -- used online for every incoming job -- and the social welfare of an optimal solution. In~\cite{Zhang17a}, the competitive ratio depends on the worst value of some job characteristic.

\section{System Description}
\label{sec.model}

In this section, we describe the QoS-differentiated cloud pricing system and derive the corresponding internal resource management model.

\subsection{The Service and Pricing Model}\label{sec_pricing}
\label{sec.pricing-model}


\vspace{0.25em}\noindent\textbf{Goal of the CSP.} A cloud service provider (CSP) holds a fixed computing capacity of $m$ servers. Users of different sources have diverse latency requirements, which could play a significant role in improving resource utilization.
Pricing them properly could incentivize users to express such diversity in latency requirements and further allows the CSP to take advantage of such diversity to do QoS-differentiated resource management, thus improving resource efficiency. A main aim of this work is to satisfy jobs' QoS requirements while maximizing the CSP's revenue by increasing its resource efficiency.

\vspace{0.25em}\noindent\textbf{SLAs.} There are $L$ {\em Service Level Agreements} (SLAs). In general, a SLA is specified by the value of some job characteristic: in this work we focus on the job's waiting times, namely, $\varphi_{1}$, $\varphi_{2}$, $\cdots$, $\varphi_{L}$. When a user $i$ selects the $l$-th SLA, its job $j$ is expected to be completed by a time $a_{j}+s_{j}+\varphi_{l}$ where $a_{j}$ and $s_{j}$ are the arrival time and the runtime of $j$, respectively. To be clear, after a job arrives at the computing system under the $l$-th SLA, it is either waiting in some queue or being served; the expected total waiting time that it spends in the system is $\varphi_{l}$.

\vspace{0.25em}\noindent\textbf{QoS-Differentiated Pricing.} The price of the $l$-th SLA is denoted by $\theta_{l}$. Here, the parameters $\varphi_{l}$ and $\theta_{l}$ are such that (\rmnum{1}) $0\leq \varphi_{1} < \cdots < \varphi_{L}$, and (\rmnum{2}) $\theta_{1} > \cdots > \theta_{L} > 0$.
The first SLA specified by $\{\varphi_{1}, \theta_{1}\}$ corresponds to the standard usage-based pricing in cloud markets (e.g., on-demand instances in Amazon EC2, the pricing in Google cloud). In such SLA, users do not have incentives to tolerate significant delays in completing their jobs, i.e., $\phi_1$ is set to a small and negligible value. The corresponding value $\theta_{1}$ is the price of utilizing a unit of computing resource, i.e., a server or a virtual machine, in the unit of time for such SLA. For the other SLAs, prices are lower than $\theta_{1}$, at the expense of delaying the completion of tagged job $j$ to an expected time larger than $a_{j}+ s_{j}$. The larger the delay a user can tolerate, the lower the price.

The overall pricing and scheduling framework developed in the rest of the paper is simply illustrated in Figure~\ref{Fig.4}.


\begin{figure}[h]
  \centering
  \includegraphics[width=3.25in]{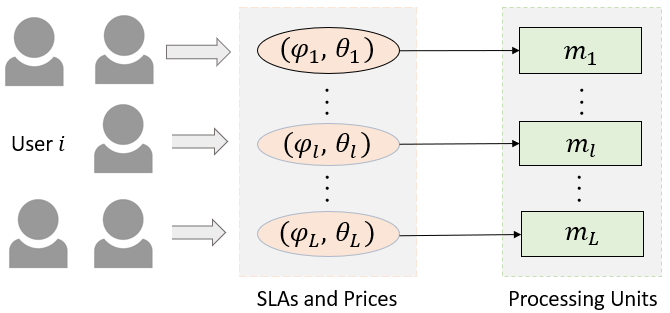}
  \caption{A cloud pricing system: when a user $i$ selects the $l$-th SLA, all its jobs are processed at the $l$-th processing unit and are expected to have a waiting time $\varphi_{l}$ and it will be charged a price $\theta_{l}$ per second in order to use a computing unit. The CSP possesses $m$ servers and the $l$-th processing unit is assigned $m_{l}$ servers where $\sum_{l=1}^{L}{m_{l}}\leq m$.}\label{Fig.4}
\end{figure}

\subsection{User's Choice}
\label{sec.user-choice}

We assume that there are a total of $K$ users, denoted by a set of indexes $\mathcal{C}=\{1, 2, \cdots, K\}$.  As described by the prospect theory of Tversky and Kahneman \cite{Tversky92a} and subsequent behavioral experiments \cite{Leclerc95a}, user's attitude towards delay could be characterized as follows. Users become more and more impatient up to some delay threshold, but increasingly insensitive once the threshold passes. Furthermore, various users may obtain different values upon completion of their jobs; thus every user is associated with a particular weight. Formally, the utility (or willingness-to-pay) is modeled by a class of functions $\mathcal{U}(\varphi, \alpha)$, where $\varphi$ represents the delay of completing a user's job; also, users are differentiated by associating a different weight $\alpha$ to each user.
As a result, we define a class of such utility functions that are widely used in cloud, electric grids, and other network pricing systems.
\begin{assumption}\label{property-1}
The user utility $\mathcal{U}(\varphi, \alpha)=\alpha\cdot \mathcal{U}^{\prime}(\varphi) + \mathcal{U}^{\prime\prime}(\varphi)$ satisfies the following properties:
\begin{itemize}
\item [(\rmnum{1})] Decreasing in the delay: users' utility decreases as the delay $\varphi$ increases where $\frac{\partial \mathcal{U}^{\prime}}{\partial \varphi} < 0$ and $\frac{\partial \mathcal{U}^{\prime\prime}}{\partial \varphi} < 0$ for smooth utility functions;
\item [(\rmnum{2})] Convexity in the delay: the marginal value is non-decreasing in $\varphi$ where $\frac{\partial^{2} \mathcal{U}}{\partial \varphi^2} \geq 0$ for smooth utility functions;

\item [(\rmnum{3})] Increasing in $\alpha$: fixing the value of $\varphi$, $\mathcal{U}(\varphi, \alpha)$ is an increasing linear function of $\alpha$ where $\mathcal{U}^{\prime}(\varphi)>0$.
\end{itemize}
\end{assumption}

This class of functions covers a large range of utility functions.
For example, it is often instantiated as quadratic utility functions in electric markets \cite{Samadi10a}. The analogy of cloud and electric markets as utility models lies in that they both provide users with on-demand access to shared resources (virtualized machines, electricity) that are paid according to consumption \cite{Brynjolfsson10a}.
In both cloud computing~\cite{Xu13a,Zheng16a,Zheng15} and network systems \cite{Mahindra14a,Chiang07a} research communities, utility functions are also specified as the product of a weight, namely $\alpha$, and a logarithmic function \cite{Zheng16a,Zheng15,Mahindra14a} ({\em resp.} an exponential function \cite{Xu13a}). Examples of such utility functions $\mathcal{U}(\varphi, \alpha)$ are:
\begin{center}
$(\psi+\varphi)^{-1}+\alpha\cdot ((\psi+\varphi)^{-1})^{2}$,\;\;  $\alpha\cdot \log\left( 1+(\psi+\varphi)^{-1} \right)$,
\end{center}
or, 
\begin{equation}\label{equa-uf1}
\alpha\cdot (1-\beta)^{-1}\cdot (\psi+\varphi)^{\beta-1},
\end{equation}
where $\beta\in (0, 1)$ and $\psi\in (0, \infty)$. 

In the rest of the work, we incorporate each user preference, namely the dependence on $\alpha$, directly into the users' utility function as follows.
\begin{definition}\label{def-utility}
Upon completion of a job whose waiting time in the system is $\varphi$, the utility of user $i$ writes
\begin{equation}
\mathcal{U}_{i}(\varphi) = \mathcal{U}(\varphi, \alpha_{i}), \enskip \varphi\geq 0
\end{equation}
where $\alpha_{i}$ is the weight of user $i$, whereas $\mathcal{U}(\varphi, \alpha)$ satisfies Assumption~\ref{property-1}.
\end{definition}

Although $\varphi$ is the expected job waiting time of a user, its utility is in fact measured in dollars per second, spent to use a given server dedicated to the $l$-th SLA: this corresponds to a pay-per usage pricing model under QoS-differentiated service. The smaller the job waiting time, the larger its (unit) utility. For example, the utility of user $i$ for completing a job $j$, of runtime $s_{j}$, under the $l$-th SLA, is $s_{j}\cdot \mathcal{U}_{i}(\varphi_{l})$. 

\begin{figure}[h]
  \centering
  \includegraphics[width=3.1in]{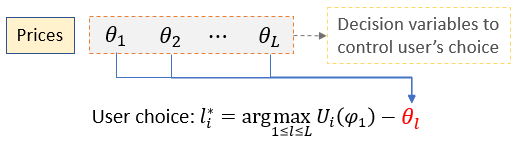}
  \caption{Decision and job dispatching process at the $l$-th processing unit.}\label{Fig.6}
\end{figure}

\vspace{0.2em}\noindent\textbf{Decision-making.} Each user $i$ will select some SLA to maximize its (unit) surplus, i.e., its utility minus the paid price under some SLA. Formally, the surplus of user $i$ under the $l$-th SLA is defined as $\mathcal{U}_{i}(\varphi_{l}) - \theta_{l}$; user $i$ will select the $l_{i}^{*}$-th SLA under which its surplus is greater than zero and
\begin{equation}\label{equa-choice}
l_{i}^{*}=\operatorname*{arg\,max}_{1\leq l\leq L}\{\mathcal{U}_{i}(\varphi_{l}) - \theta_{l}\}.
\end{equation}
In particular, if there exist multiple SLAs under which user $i$ achieves the same maximum surplus, user $i$ will select one of the multiple SLAs randomly.
The decision-making process of user $i$ is also illustrated in Fig.~\ref{Fig.6}; the prices could be used as decision variables to control which SLA will be chosen by user $i$, resulting in that its choice of latency requirements fulfills the CSP's expectation. Hence, by setting the SLA prices properly, the pricing system introduced in Sec.~\ref{sec_pricing} could provide incentives for users to express their diversity in latency requirements, and the CSP could improve the resource efficiency in terms of revenue after a corresponding resource management scheme is proposed to complement the pricing model. 

\subsection{Capacity Partition}

\vspace{0.2em}\noindent\textbf{Job arrival at every SLA.} Given the SLA prices of our pricing model, each user $i\in\mathcal{C}$ will select some SLA $l$ based on (\ref{equa-choice}) and all its jobs are expected to be completed with an expected waiting time no greater than $\varphi_{l}$.
User $i$ continuously submits jobs to the cloud system over time. As done in previous literature for IaaS clouds~\cite{Zheng16a,Dierks16a,Bruneo14a,Chang16a}, we assume that the job arrival rate of user $i$ follows a Poisson distribution with mean $h_{i}$. Also, job arrival processes of different users are assumed independent. Thus, inter-arrival time between two consecutive jobs are exponentially distributed. We denote by $\mathcal{C}_{l}$ the subset of users who choose the $l$-th SLA; it holds that $\cup_{l=1}^{L} {\mathcal{C}_{l}} =\mathcal{C}$, where $\mathcal{C}_{l}$ may or may not be empty for some SLA $l$. We denote by $\Lambda_{l}$ the total job arrival rate of users at the $l$-th processing unit, i.e.,
\begin{equation}\label{equa-rate}
\Lambda_{l} = \sum\nolimits_{j\in\mathcal{C}_{l}}{h_{j}},
\end{equation}

\vspace{0.2em}\noindent\textbf{Capacity partition.} The CSP holds $m$ servers and the whole computing capacity is divided into $L$ processing units, as illustrated in Fig.~\ref{Fig.4}; for all $1\leq l\leq L$, the $l$-th processing unit is assigned $m_{l}$ servers where
\begin{equation}\label{equa-num-machine}
\sum\nolimits_{l=1}^{L}{m_{l}} \leq m.
\end{equation}
The $l$-th processing unit is used to exclusively process all jobs that arrive for the $l$-th SLA and each server will process the jobs with the same QoS requirement.

In this paper, we consider the case that the CSP possesses a fixed number of servers that are assumed homogeneous virtual machines and the CSP has a limited capacity to process the arriving jobs. So, a admission control scheme is needed to determine the rate of accepting jobs that are  processed at the $l$-th processing unit; here, some of the arriving jobs may not be processed in order to maximize the revenue. As elaborated later, the value of $m_{l}$ determines the processing capacity of the $l$-th unit, i.e., the maximum possible rate of accepting jobs in order not to violate the $l$-th SLA. Accordingly, we denote by $\Lambda_{l}^{+}$ the rate of accepting jobs for processing at the $l$-th processing unit, where
\begin{equation}\label{equa-constraint-1}
\Lambda_{l}^{+}\leq \Lambda_{l},
\end{equation}

The runtimes of jobs are assumed to be independent and identically distributed random variables, following a general distribution with mean $\omega$. In queuing theory, the mean $\omega$ is often normalized to be one for the sake of analysis and we also do so to make the existing results for job dispatching directly accessible, as seen in Section~\ref{sec.dispatching-policies}. The mean service rate is the average number of jobs that a server can serve per unit time, denoted by $\mu$ where $\mu = 1/\omega=1$. We denote by $W_{l}$ the total workload processed by the $l$-th processing unit in a unit of time and it is the product of the rate of accepting and processing jobs and the mean runtime of jobs, i.e.,
\begin{equation}\label{equa-aggregation}
W_{l} = \Lambda_{l}^{+}\cdot \omega = \Lambda_{l}^{+}.
\end{equation}
What happens at the $l$-th processing unit is also illustrated in Fig.~\ref{Fig.5} and will be further elaborated in the following. As seen later, the decision variables for revenue maximization in the resource management part are $\{\Lambda^{+}\}_l$s and $\{m\}_{l}$s.

\begin{figure*}[h]
  \centering
  \includegraphics[width=6.6in]{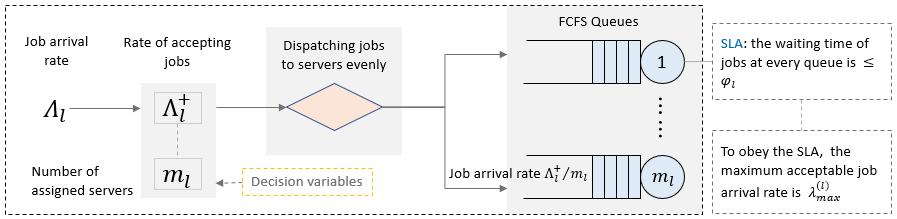}
  \caption{Decision and job dispatching process at the $l$-th processing unit.}\label{Fig.5}
\end{figure*}

\subsection{Job Assignment}\label{sec.resource-management}


All accepted jobs to be processed at the $l$-th unit are dispatched to each of its servers under a certain dispatching policy. After dispatching, the first-come-first-served (FCFS) discipline is applied to process jobs arriving at a single server~\cite{Dierks16a,Zheng16a,Bruneo14a,Chang16a}. The dispatching policy determines (\rmnum{1}) how many of the accepted jobs will be assigned to each server of the $l$-th processing unit, and (\rmnum{2}) the specific relation between the job arrival rate at each server and the expected job waiting time. In our analysis of the QoS-differentiated architecture, the thing that matters is that these two features determine the processing capacity of the $l$-th unit, i.e., its maximum possible rate of accepting and processing jobs.

In this paper, we consider a class of policies, denoting all such policies by a set $\mathcal{D}$, with the following properties:
\begin{property}[Load Balancing]\label{lemma-dispatching}
Suppose that a policy in $\mathcal{D}$ is applied to every processing unit. Then, the job arrival rate at each server of the $l$-unit follows a Poisson distribution with mean $\lambda_{l} = \Lambda_{l}^{+}/m_{l}$.
\end{property}

The specific policies contained in $\mathcal{D}$ are introduced in Section~\ref{sec.dispatching-policies}. Property~\ref{lemma-dispatching} says that, under a policy in $\mathcal{D}$, the loads among the $m_{l}$ servers of the $l$-th unit is uniformly balanced, and every server of the same processing unit could be viewed as a single queue with the identical job arrival rate $\lambda_{l}=\Lambda_{l}^{+}/m_{l}$.

The expected waiting time of jobs increases with the job arrival rate at every server and conversely the following property holds:

\begin{property}\label{def-arrival}
Let's consider a policy in $\mathcal{D}$, and a resulting job assignment: the job arrival rate $\lambda$ at every server of some processing unit is an increasing function of the expected waiting time of jobs $\varphi$, irrelevant of the number of servers assigned to that unit.
\begin{equation}\label{fun-arrival}
\lambda = \mathcal{Q}(\varphi).
\end{equation}
\end{property}


If the expected waiting time of jobs at a server is $\leq \varphi_{l}$, the job arrival rate at this server should be $\leq \lambda = \mathcal{Q}(\varphi)$. The function $\mathcal{Q}(\varphi)$ in Property~\ref{def-arrival} will be instantiated when we introduce the specific dispatching policies in Section~\ref{sec.dispatching-policies}.

In this paper, the $l$-th SLA guarantees that the expected waiting time of jobs is no greater than $\varphi_{l}$ at the $l$-th processing unit. This is achieved by guaranteeing that the waiting time of processed jobs at every server is no greater than $\varphi_{l}$. So, every server has a processing capacity constraint and the below lemma follows by Property~\ref{def-arrival}.

\begin{lemma}\label{lemma-max-single-rate}
Under a job assignment policy in $\mathcal{D}$, at a single server of the $l$-th unit, the maximum acceptable job arrival rate under the $l$-th SLA is $\mathcal{Q}(\varphi_{l})$, denoted by $\lambda_{max}^{(l)}$.
\end{lemma}

The relation of the job arrival rate at every server and the total rate of processing jobs at the $l$-th unit is described in Property~\ref{lemma-dispatching}, and we thus have the following lemma.

\begin{lemma}\label{lemma-max-rate}
Given an arbitrary number $m_{l}$, suppose that the $l$-th unit is assigned $m_{l}$ servers. The maximum possible rate of accepting and processing jobs at this unit is $\lambda_{max}^{(l)}\cdot m_{l}$, i.e., linearly proportional to $m_{l}$, in order not to violate the $l$-th SLA.
\end{lemma}

The rate $\Lambda_{l}^{+}$ of accepting jobs is constrained by the processing capacity of the $l$-th unit and we have by Lemma~\ref{lemma-max-rate} that:
\begin{equation}\label{constraint-3}
\Lambda_{l}^{+} \leq m_{l}\cdot \lambda_{max}^{(l)}.
\end{equation}

\subsubsection{Dispatching Policies}
\label{sec.dispatching-policies}

In IaaS services, dispatching policies should have the capacity of simultaneously dispatching numerous jobs to servers in parallel, rather than sequentially like what a central queue does \cite{Ousterhout13a,Rasley16a}. In this paper, we consider three dispatching policies that not only satisfy Properties~\ref{lemma-dispatching} and~\ref{def-arrival} but also are promising in cloud environments \cite{Zheng16a,Wickremasinghe10a,Ousterhout13a,Jennings15a}. Each policy has a particular way to communicate among jobs, dispatchers and servers and the three policies are described below.

\vspace{0.15em}\noindent\textbf{Random}: for every job that arrives at the $l$-th processing unit, immediately choose a server with the same probability $\frac{1}{m_{l}}$ and assign this job to the server;

\vspace{0.15em}\noindent\textbf{PTC}: for every arriving job, choose two servers randomly, probe them, and, assign this job to the server with fewer queued jobs;

\vspace{0.15em}\noindent\textbf{JIQ}-Random: at the $l$-th unit, $\kappa_{l}$ dispatchers are maintained. Whenever a server of the $l$-th unit becomes idle, it randomly chooses a dispatcher and informs the dispatcher of its idleness. On the other hand, for every arriving job, immediately assign it to a dispatcher chosen randomly; if there are recorded idle servers at this dispatcher, it will assign the job to an idle server chosen randomly; otherwise, just assign the job to a randomly chosen (idle or non-idle) server from the $l$-th unit. For the analysis sake, we shall consider the case when $\phi=m_{l}/\kappa_{l}$ is constant across the deployment.

Now, we instantiate $\mathcal{D}$ as a set of the three dispatching policies: Random, PTC, and JIQ-Random, and give the related results to also instantiate the function (\ref{fun-arrival}) in Property~\ref{def-arrival} when a particular policy in $\mathcal{D}$ is considered. We note that, in the literature of cloud computing, Pareto distribution is often used for characterizing the job's runtime of a published Google's workload \cite{Zheng16a}, while exponential distribution is also applied commonly \cite{Bruneo14a,Chang16a} or for a particular Microsoft's production workload \cite{Rasley16a}. Often, the mean of runtimes is normalized to be one for simplifying the analysis, as seen in the analysis of the PTC and JIQ-Random policies \cite{Mitzenmacher01a,Lu11a}.


Every policy in $\mathcal{D}$ satisfies Property~\ref{lemma-dispatching} and the reason for this is its randomness when choosing a server or dispatcher and the details can be found in \cite{Zheng16a,Mitzenmacher01a,Lu11a}. Now, we introduce the related results for Property~\ref{def-arrival}. In a study of cloud brokerage service \cite{Zheng16a}, the Random policy is used. When the runtimes of jobs follow an exponential distribution, all jobs have an expected waiting time
\begin{equation}\label{equa-exponential-1}
\varphi = \rho/(\mu - \lambda).
\end{equation}
In the case that the runtimes of jobs follow a type of heavy-tailed distribution, i.e., Pareto distribution, with a shape parameter $\eta>1$ and a scale parameter $x_{m}$ (i.e., the minimum runtime), the expected runtime of jobs is $\frac{\eta x_{m}}{\eta-1}$. When $\frac{1}{\lambda}-\frac{\eta}{\eta-1}x_{m}>0$, all jobs have a finite expected waiting time \cite{Zheng16a}
\begin{equation}\label{equa-Pareto-1}
\varphi = \frac{1}{\lambda}\log\left( \frac{ \eta\cdot \left( \lambda x_{m} \right)^{\eta} \cdot \Gamma(-\eta, \lambda x_{m})  }{ 1 - \lambda\cdot (\eta x_{m}) / (\eta - 1)} \right),
\end{equation}
where $\Gamma(s, z) = \int_{z}^{\infty}{x^{s-1}e^{-x}}dx$ is the upper incomplete gamma function.
In \cite{Mitzenmacher01a}, the PTC policy is considered and the analytical result is given by assuming the job's runtime follows an exponential distribution. The expected waiting time of jobs is
\begin{equation}\label{equa-ptc-exp}
 \varphi = \lambda^{2}/\left(1-\lambda^{2}\right).
\end{equation}
In \cite{Lu11a}, the JIQ policy is considered and a general distribution is used to characterize the runtime of jobs. The expected waiting time of jobs under JIQ is
\begin{equation}\label{equa-JIQ}
\varphi = \left( (C+1)/2 \right) \cdot \left( \lambda/(1-\lambda)/(1+\phi) \right),
\end{equation}
where the parameter $C$ is the ratio of the variance of the particular runtime distribution to its mean squared.
the mean runtime of jobs is 1.

With (\ref{equa-exponential-1}), (\ref{equa-Pareto-1}), (\ref{equa-ptc-exp}), and (\ref{equa-JIQ}), we can illustrate how the utilization increases with waiting time where the utilization is $\lambda\cdot\omega$, i.e., the product of the mean runtime of jobs and the job arrival rate at a server. A graph is reported in Figure~\ref{Fig.12}, for $\mu=1$ and $\phi=10$. The red dotted, blue, and green lines, respectively, represent the performance of Random, PTC, and JIQ policies in the case when the service time follows an exponential distribution; the red line represents the curve of the Random policy where the service time follows a Pareto distribution with $x_{m}=1.001/2.001$ and $\eta=2.001$.

Finally, the key notation of this paper is listed in Table~1. In this paper, we use the script of a capital letter to denote a set, e.g., $\mathcal{C}$ denotes a set of the indexes of users. In contrast, the script of a capital letter with input variables is often used to denote a function, e.g., $\mathcal{U}(\varphi, \alpha)$ denotes the utility function. In addition, as seen later, the bold letter is used to denote a vector, e.g., $\boldsymbol{\theta}$ denotes the SLA prices $(\theta_{1},$ $\cdots,$ $\theta_{L})$.

In this paper, all omitted proofs can be found in the appendix.

\begin{figure}
\begin{center}
  \includegraphics[width=2.95in]{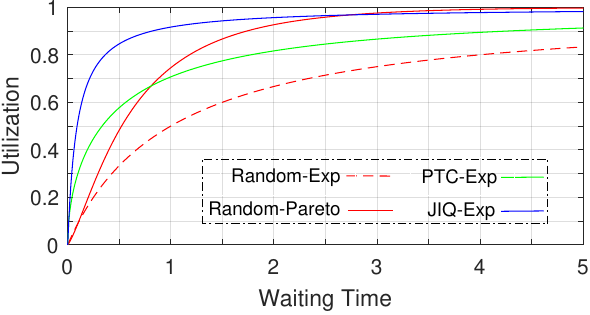}
  \caption{The increase of system utilization with the waiting time.}\label{Fig.12}
\end{center}
\end{figure}

\begin{table}[t]
	\centering
	\begin{threeparttable}[t]
       \label{table-notation}
		\caption{Key Notation}
		\begin{tabular}{| C{2cm} | p{6cm} |}   
			\hline
			{\bf Symbol} & {\bf Explanation}\\ \hline
			$L$ & the number of SLAs \\ \hline

            $\varphi_{l}$ & the waiting time of the $l$-th SLA \\ \hline

            $\theta_{l}$ ({\em resp.} $\theta_{l}^{*}$) & the price ({\em resp.} the optimal price) of the $l$-th SLA \\ \hline

			$m$ & the total number of servers possessed by a CSP \\ \hline

			$m_{l}$ ({\em resp.} $m_{l}^{*}$) & the ({\em resp.} optimal) number of servers assigned to the $l$-th processing unit \\ \hline


			$\Lambda_{l}^{+}$ ({\em resp.} $\Lambda_{l}^{*}$) & the ({\em resp.} optimal) rate of accepting job arrival rate at the $l$-th processing unit \\ \hline


			
			$\lambda_{max}^{(l)}$ & the maximum acceptable job arrival rate at each server of the $l$-th SLA processing unit, in order not to violate the $l$-th SLA \\ \hline

		\end{tabular}

		\label{table}
	\end{threeparttable}
\end{table}


\section{Goals and Problem Formulation}\label{sec.system-optimizations}


After we introduced the QoS-differentiated pricing and resource management architecture, we proceed to tackle the main objective of this work, i.e., the maximization of a CSP's revenue. To achieve this, the whole optimization process is divided into two correlated phases: {\em capacity planning} and {\em optimal prices determination}.

\subsection{Decision Variables, and Objective}

We denote the prices of the 1-th, 2-th, $\cdots$, $L$-th SLAs by a price vector $\boldsymbol{\theta}$ $=$ $(\theta_{1},$ $\cdots,$ $\theta_{L})$, the number of servers assigned to every processing unit $l$ by $\boldsymbol{m}=(m_{1}, \cdots, m_{L})$, and the rate of accepting and processing jobs at every unit $l$ by $\boldsymbol{\Lambda^{+}} =$ $(\Lambda^{+}_{1}, \cdots,$ $\Lambda^{+}_{L})$. To maximize the revenue generated by the QoS-differentiated architecture defined in Section~\ref{sec.model}, the decision variables are $\boldsymbol{m}$, $\boldsymbol{\Lambda^{+}}$, and $\boldsymbol{\theta}$, which are elaborated in the following.

Recall the pricing model and the user's choice process in Section~\ref{sec.pricing-model}~and~\ref{sec.user-choice}. As illustrated in Fig.~\ref{Fig.6}, each tagged user selects a SLA based on \eqref{equa-choice} to maximize its surplus; as a result, for all $l\in [1, L]$, the subset of users $\mathcal{C}_{l}$ who choose the $l$-th SLA is determined. Further, the arrival rate $\Lambda_{l}$ of jobs at every processing unit $l$ is obtained by \eqref{equa-rate}. Here, the users' choices $\mathcal{C}_{1}$, $\cdots$, $\mathcal{C}_{L}$ and the corresponding job arrival rates $\Lambda_{1}$, $\cdots$, $\Lambda_{L}$ are dominated by the price vector $\boldsymbol{\theta}$, which is a decision variable.

Further, the CSP has a fixed number of servers. Under every SLA, there is a threshold such that the SLA will be violated if the job arrival rate at a server exceeds the threshold by Lemma~\ref{lemma-max-single-rate}; thus, the CSP has limited processing capacity. Facing the jobs arriving at every processing unit, the decision variables are the rate of accepting and processing jobs at every processing unit, i.e., $\boldsymbol{\Lambda^{+}}$, and, the number of servers to assigned to every unit, i.e., $\boldsymbol{m}$; intuitively, our aim is to accept as many jobs as possible at one SLA under which every server of some unit could yield higher revenue.

The total workload processed by the $l$-th unit in a unit of time is $W_{l}$ defined in (\ref{equa-aggregation}). Users are charged for each unit of resources they actually use (i.e., the actual use of a server in a unit of time). For a unit of time, the total revenue obtained from the whole $l$-th unit is $W_{l}\cdot \theta_{l}$. The CSP's total revenue is the sum of the revenues of all units, i.e.,
\begin{equation}\label{equa-max}
\mathcal{V}(\boldsymbol{m}, \boldsymbol{\Lambda^{+}}, \boldsymbol{\theta})  = \sum\nolimits_{l=1}^{L}{W_{l}\cdot \theta_{l}} \overset{(b)}{=} \sum\nolimits_{l=1}^{L}{\theta_{l}\cdot \Lambda_{l}^{+}}
\end{equation}
where the equation (b) is due to (\ref{equa-aggregation}). Here $\mathcal{V}(\boldsymbol{m}, \boldsymbol{\Lambda^{+}}, \boldsymbol{\theta})$ can be interpreted as the expected total revenue generated by the whole cloud system in a unit of time.


\subsection{Capacity Planning: $\boldsymbol{m}$ and $\boldsymbol{\Lambda^{+}}$}
\label{sec.prob-capacity-planning}

We first ignore the optimality of the SLA prices $\boldsymbol{\theta}$ and consider the case where $\boldsymbol{\theta}$ is pre-assigned arbitrarily; our aim is to maximize (\ref{equa-max}), i.e., the revenue of a CSP. In this case, our decision variables are the rate of accepting jobs $\Lambda_{l}^{+}$ at every unit $l$ and the number $m_{l}$ of assigned servers; the job arrival rate $\Lambda_{l}$ at every unit $l$ is known. 

\begin{problem}\label{prob1}
Under arbitrary, pre-assigned, SLA's prices $\boldsymbol{\theta}$, we denote $\mathcal{V}(\boldsymbol{m}, \boldsymbol{\Lambda^{+}}, \boldsymbol{\theta})$ in (\ref{equa-max}) by $\mathcal{V}_{\boldsymbol{\theta}}(\boldsymbol{m}, \boldsymbol{\Lambda^{+}})$; the revenue maximization problem can be written as
\[
\mbox{maximize} \; \mathcal{V}_{\boldsymbol{\theta}}(\boldsymbol{m}, \boldsymbol{\Lambda^{+}})
\]
with $\boldsymbol{m}$ and $\boldsymbol{\Lambda^{+}}$ as decision variables and subject to constraints \eqref{equa-num-machine}, \eqref{equa-constraint-1}, and \eqref{constraint-3}.
\end{problem}


Constraint (\ref{equa-num-machine}) says that the total number of servers assigned to different processing units is $\leq$ the total number of servers available. (\ref{equa-constraint-1}) says that at the $l$-th unit the rate of accepting and processing jobs is $\leq$ the total job arrival rate of users who choose the $l$-th SLA, which is determined by the SLA prices. (\ref{constraint-3}) says that the rate of accepting jobs should be $\leq$ the maximum possible rate of accepting jobs in order not to violate the $l$-th SLA, by Lemma~\ref{lemma-max-rate}.

We will in Section~\ref{sec.capacity-planning} give an optimal solution to Problem~\ref{prob1}, i.e., the optimal values of $\boldsymbol{m}$ and $\boldsymbol{\Lambda^{+}}$, denoted by $\boldsymbol{m^{*}}$ and $\boldsymbol{\lambda^{*}}$; as a result, we could derive the maximum
revenue of a CSP under an arbitrary, pre-assigned, tuple $\boldsymbol{\theta}$ of SLA prices. Here, the optimal capacity planning and admission control are realized and defined by $\boldsymbol{m^{*}}$ and $\boldsymbol{\Lambda^{*}}$, which are linked to an arbitrarily given $\boldsymbol{\theta}$; thus, $\boldsymbol{m^{*}}$ and $\boldsymbol{\Lambda^{*}}$ could also be viewed as the functions of $\boldsymbol{\theta}$: $\boldsymbol{m}_{\theta}^{*}$ and $\boldsymbol{\lambda}_{\theta}^{*}$ (or $\boldsymbol{m^{*}(\theta)}$ and $\boldsymbol{\lambda^{*}(\theta)}$).


\subsection{Optimal Prices: $\boldsymbol{\theta}$}\label{sec.prob-2}

Under posted price models, users are price-takers who choose a SLA according to~\eqref{equa-choice}. The determination of optimal SLA prices $\boldsymbol{\theta}$ is needed to steer each user to the most suitable SLA for revenue maximization. 
Once Problem \ref{prob1} is solved optimally, our objective function (\ref{equa-max}) could be expressed as a function of $\boldsymbol{\theta}$:
\begin{align}\label{fun-optimal-prices}
\mathcal{V}\left( \boldsymbol{m}, \boldsymbol{\Lambda^{+}}, \boldsymbol{\theta} \right) = \mathcal{V}\left( \boldsymbol{m(\theta)}, \boldsymbol{\Lambda^{+}(\theta)}, \boldsymbol{\theta} \right);
\end{align}
this means that, currently, we could derive the CSP's maximum revenue once the optimal SLA prices are derived; here, the decision variable is the price vector $\boldsymbol{\theta}$. Let $\Theta$ denote all possible tuples of SLA prices, i.e., $\Theta = \{\boldsymbol{\theta}\, |\, \theta_{1}, \cdots, \theta_{L}\in \mathcal{R}^{+} \}$.
The optimal pricing problem writes:
\begin{problem}\label{prob2}
The optimal price vector is defined as
\begin{equation}\label{equa-pp}
  \boldsymbol{\theta}^{*}:=\arg\max\nolimits_{\boldsymbol{\theta}\in\Theta} \; \text{function}\; (\ref{fun-optimal-prices}),
\end{equation}
subject to the condition \eqref{equa-choice}, i.e., how the job rate for each SLA is affected by price vector $\boldsymbol{\theta}$.
\end{problem}

Recall that there are $K$ users and $\mathcal{C}=\{1, 2, \cdots, K\}$. Under an arbitrary tuple of SLA prices $\boldsymbol{\theta}\in\Theta$, the $l$-th SLA is chosen by a subset of users $\mathcal{C}_{l}\subseteq\mathcal{C}$. In the case where $\mathcal{C}_{l}$ is empty, no users will choose the $l$-th SLA; in turn, this price could be equivalently viewed to be $\infty$. Based on such observations, we consider a corresponding subset $\mathcal{S} = \{  l_{1}, l_{2}, \cdots,$ $l_{L^{\prime}}  \}$ of $\{1, 2, \cdots, L\}$, and let $\mathcal{S}$ define a subset of SLAs to be chosen by users: each of the $l_{1}$-th, $\cdots$, $l_{L^{\prime}}$-th SLAs will be chosen by some users while no users chooses the other SLAs, that is, $\mathcal{C}_{l}\neq \emptyset$ for all $l\in\mathcal{S}$ and $\mathcal{C}_{l}=\emptyset$ for all $l\in\{1, 2, \cdots, L\}-\mathcal{S}$. We denote by $\Theta_{\mathcal{S}}$ all price tuples such that, under every tuple of SLA prices $\theta\in\Theta_{\mathcal{S}}$, the subset of SLAs chosen by users is defined by the subset $\mathcal{S}$. As a result, all possible tuples of SLA prices are grouped by users' choices of SLAs and each group is distinguished by the $\mathcal{S}$ defined above, denoted by $\Theta_{\mathcal{S}}$, that is,
\begin{center}
$\Theta = \bigcup_{\mathcal{S}\subseteq\{1, 2, \cdots, L\}}{\Theta_{\mathcal{S}}}$
\end{center}

Our subsequent results in Section~\ref{sec.optimal-prices} show that only if we are given that subset $\mathcal{S}$ of $\{1, 2, \cdots, L\}$ that is used to define which SLAs to be chosen by users as explained above, an analysis of the structural properties of user's choices could derive the corresponding optimal prices of SLAs, denoted by $\boldsymbol{\theta^{*}(\mathcal{S})}$, that maximizes (\ref{fun-optimal-prices}); in other words, among all $\boldsymbol{\theta}\in\Theta_{\mathcal{S}}$, (\ref{fun-optimal-prices}) achieves the highest value when $\boldsymbol{\theta} = \boldsymbol{\theta^{*}(\mathcal{S})}$, that is,
\begin{equation}\label{equa-pp-1}
  \boldsymbol{\theta^{*}(\mathcal{S})} := \arg\max\nolimits_{\boldsymbol{\theta}\in \Theta_{\mathcal{S}}}\; (\ref{fun-optimal-prices}).
\end{equation}
Here, the prices of the $l_{1}$-th, $\cdots$, $l_{L^{\prime}}$-th SLAs are optimally determined by the analysis and the prices of the other SLAs are trivially set to $\infty$ since no user chooses them. To finally maximize (\ref{fun-optimal-prices}), we only need to select such a subset $\mathcal{S}^{*}$ from all the subsets of $\{1, 2, \cdots, L\}$ that (\ref{fun-optimal-prices}) generates the highest value under $\boldsymbol{\theta^{*}(\mathcal{S}^{*})}$; formally, we simply denote the objective function $\mathcal{V}\left( \boldsymbol{m(\theta)}, \boldsymbol{\Lambda^{+}(\theta)}, \boldsymbol{\theta} \right)$ in (\ref{fun-optimal-prices}) by $\mathcal{G}(\boldsymbol{\theta})$ and have that
\begin{center}
$\mathcal{S}^{*} := \arg\max_{\mathcal{S}\subseteq\{1, 2, \cdots, L\}}\, \mathcal{G}(\boldsymbol{\theta^{*}(\mathcal{S})}) $.
\end{center}
Here, $\boldsymbol{\theta^{*}(\mathcal{S}^{*})}$ is an optimal solution to Problem~\ref{prob2}.

In the above process, the only question is how to derive the optimal SLA prices $\boldsymbol{\theta^{*}(\mathcal{S})}$, and we will in Section~\ref{sec.optimal-prices} solve (\ref{equa-pp-1}) optimally.

\section{Revenue Maximization}\label{sec.optimization}

In this section, we maximize the revenue of a CSP given its capacity, optimally solving the two problems in Section~\ref{sec.system-optimizations}.
\subsection{Capacity Planning}
\label{sec.capacity-planning}

In this subsection, we give an optimal solution to Problem~\ref{prob1}, producing an optimal capacity planning scheme in the case that the SLA prices are pre-assigned.

Under pre-assigned SLA prices, the job arrival rate of users who choose the $l$-th SLA is $\Lambda_{l}$; we have by Lemma~\ref{lemma-max-rate} that the minimum number of servers needed to process all the arriving jobs at the $l$-th processing unit is $\lceil \Lambda_{l}/\lambda_{max}^{(l)} \rceil$. In the case that $\sum_{l=1}^{L}{\lceil \Lambda_{l}/\lambda_{max}^{(l)} \rceil}\leq m$, there are enough servers to process all jobs arriving at different units. In this case, an optimal solution to Problem~\ref{prob1} is as follows: $\Lambda_{l}^{+}=\Lambda_{l}$ and $m_{l} = \lceil \Lambda_{l}/\lambda_{max}^{(l)} \rceil$; it means accepting all jobs that arrive at each processing unit and assigning $m_{l}$ servers to the $l$-th unit for all $l\in [1, L]$.

In the rest of this subsection, we study the opposite case where
\begin{equation}\label{equa-case}
\sum\nolimits_{l=1}^{L}{\lceil \Lambda_{l}/\lambda_{max}^{(l)} \rceil}> m;
\end{equation}
by Lemma~\ref{lemma-max-rate}, this means that, $m$ servers are not enough to process all jobs that arrive at different units while satisfying their QoS requirements. Here, we also note that, the following solving process itself is a purely mathematical derivation, although we sometimes additionally explain the physical meaning of some results for better perceiving their intuitions.

Let $\overline{m}_{l}$ be the maximum integer such that $\overline{m}_{l}\cdot \lambda_{max}^{(l)}\leq \Lambda_{l}$, i.e., $\overline{m}_{l} = \lfloor \Lambda_{l}/\lambda_{max}^{(l)} \rfloor$, and let $\lambda_{l}^{\prime} = \Lambda_{l} - \lambda_{max}^{(l)}\cdot \overline{m}_{l}$; here, we have
\begin{align}
& \Lambda_{l} = \overline{m}_{l}\cdot \lambda_{max}^{(l)} + \lambda_{l}^{\prime}\label{floor-1} \\
& 0\leq \lambda_{l}^{\prime}<\lambda_{max}^{(l)}\label{floor-2}.
\end{align}
\begin{lemma}\label{lemma-optimal-structure-1}
In an optimal solution to Problem~\ref{prob1}, some relation between the two decision variables $\Lambda_{l}^{+}$ and $m_{l}$ is satisfied:
\begin{enumerate}
\item if $m_{l}\cdot\lambda_{max}^{(l)}\leq \Lambda_{l}$, then $\Lambda_{l}^{+}=m_{l}\cdot\lambda_{max}^{(l)}$ where $m_{l}\leq \overline{m}_{l}$.
\item if $m_{l}\cdot\lambda_{max}^{(l)}>\Lambda_{l}$, then $\Lambda_{l}^{+}=\Lambda_{l}$ and $m_{l}=\overline{m}_{l}+1$; in this case, we also have $\lambda_{l}^{\prime} > 0$.
\end{enumerate}
\end{lemma}


The physical meaning of Lemma~\ref{lemma-optimal-structure-1} is as follows. Under the $l$-th SLA, every server has a limited processing capacity, and, given the $m_{l}$ servers at the $l$-th unit, the maximum possible rate of accepting jobs is $m_{l}\cdot \lambda_{max}^{(l)}$, by Lemma~\ref{lemma-max-rate}. In the case that the total job arrival rate of users who choose the $l$-th SLA is $\leq m_{l}\cdot \lambda_{max}^{(l)}$, the $l$-th unit will accept all the arriving jobs in order to maximize the revenue that these $m_{l}$ servers could generate in a unit of time, i.e., $\Lambda_{l}^{+}=\Lambda_{l}$; otherwise, the optimal rate of accepting jobs $\Lambda_{l}^{+}$ is $m_{l}\cdot \lambda_{max}^{(l)}$. In both the cases, every unit $l$ is assigned $m_{l}$ servers, that is, the minimum number of servers needed to process the accepted jobs under the $l$-th SLA, by Lemma~\ref{lemma-max-rate}.

In this subsection, we will finally give an optimal solution to Problem~\ref{prob1}. Next, we will combine the relation of $\Lambda_{l}^{+}$ and $m_{l}$ in Lemma~\ref{lemma-optimal-structure-1} into constraints \eqref{equa-num-machine}, \eqref{equa-constraint-1}, and (\ref{constraint-3}). We hence reduce the original problem to a series of simpler mathematical problems. This provides better insight into the structure of optimal solution to Problem~\ref{prob1}.

These reductions are mainly built on the following observation, describing the relations of the original and reduced problems. Suppose an (original) problem $A$ aims to maximize the function $\mathcal{G}(x)$, subject to $x\in\mathcal{X}$, where $\mathcal{X}$ denotes a set of all feasible solutions; an optimal solution is a $x^{*}\in\mathcal{X}$ such that $\mathcal{G}(x^{*})\geq \mathcal{G}(x)$ for all $x\in\mathcal{X}$. Given a set $\mathcal{Y}$, assume there is a mapping $f$ from $\mathcal{Y}$ to $\mathcal{X}$, i.e., $x=f(y)\in \mathcal{X}$ for all $y\in\mathcal{Y}$. The (reduced) problem $B$ aims to maximize $G(f(y))$ subject to $y\in\mathcal{Y}$.
\begin{observation}\label{lemma-reduction-AA}
Suppose this mapping $f$ is a surjection, i.e., for any $x\in\mathcal{X}$, there exists a $y\in\mathcal{Y}$ such that $x=f(y)$. An optimal solution $y^{*}$ to the problem $B$, corresponds to, an optimal solution $x^{*}$ to the problem $A$, where $x^{*}=f(y^{*})$.
\end{observation}

In Problem~\ref{prob1}, the constraints define a set of feasible solutions $\left( \boldsymbol{m},\, \boldsymbol{\Lambda^{+}} \right)$ to it and, now, we begin reducing Problem~\ref{prob1}.

\begin{observation}\label{transformation-1}
We aim to find an optimal solution to Problem~\ref{prob1}, and assume that the relation between the two variables $m_{l}$ and $\Lambda_{l}^{+}$ satisfies Lemma~\ref{lemma-optimal-structure-1}, as well as the constraints \eqref{equa-num-machine}, \eqref{equa-constraint-1}, and (\ref{constraint-3}).

Let $x_{l}$ and $y_{l}$ denote two integer variables that satisfy the following constraints:
\begin{equation}\label{constraint-4}
\begin{split}
&\hskip-2mm 0\leq x_{l}\leq \overline{m}_{l},\\
&\hskip-2mm y_{l}\in \{0, 1\} \, \text{if} \, \left( x_{l}=\overline{m}_{l} \wedge \lambda_{l}^{\prime} > 0 \right),\, \text{and} \, y_{l}=0 \, \text{otherwise},
\end{split}
\end{equation}
and,
\begin{align}
 \sum\nolimits_{l=1}^{L}{(x_{l} + y_{l})} \leq m. \label{constraint-2-1}
\end{align}
There exists a surjection from $(x_{l}, y_{l})$ to $(m_{l}, \Lambda_{l}^{+})$ as follows:
\begin{equation}\label{constraint-2}
\begin{bmatrix}
   m_{l}  &   \Lambda_{l}^{+}\\
\end{bmatrix}
=
\begin{bmatrix}
   x_{l}  &   y_{l} \\
\end{bmatrix}\times A,
\end{equation}
where $A$ is a matrix as follows:
\begin{equation*}
A =
\begin{bmatrix}
   1   & \lambda_{max}^{(l)} \\
   1   & \lambda_{l}^{\prime}
\end{bmatrix}
\end{equation*}
\end{observation}

\begin{problem}\label{prob22}
Problem~\ref{prob1} is rewritten as:
\begin{center}
maximize\enskip $\sum\nolimits_{l=1}^{L}{x_{l}\cdot (\theta_{l}\cdot \lambda_{max}^{(l)}) + y_{l}\cdot (\theta_{l}\cdot \lambda_{l}^{\prime})}$,
\end{center}
with $x_{1}, \cdots x_{L}$ and $y_{1}, \cdots y_{L}$ as decision variables and subject to the constraints \eqref{constraint-4}, and \eqref{constraint-2-1}.
\end{problem}


By Observation~\ref{transformation-1}, the relation (\ref{constraint-2}) defines a surjection from $(x_{l}, y_{l})$ to $(m_{l}, \Lambda_{l}^{+})$ in the case that the related constraints are satisfied. Due to Observation~\ref{lemma-reduction-AA}, we conclude that
\begin{lemma}[Equivalence]\label{lemma-relationship-2}
An optimal solution to Problem~\ref{prob22} corresponds to an optimal solution to Problem~\ref{prob1} where these two optimal solutions satisfy (\ref{constraint-2}).
\end{lemma}

Recall the definition of $\overline{m}_{l}$. In the case that $\overline{m}_{l}=0$ where $\Lambda_{l}<\lambda_{max}^{(l)}$, due to constraint (\ref{constraint-4}) alone, we derive that the variable $x_{l}$ is zero in any feasible solution to Problem~\ref{prob22};
We denote by $\mathcal{A}$ a subset of $\{1, 2, \cdots, L\}$ such that
\begin{enumerate}
\item [(\rmnum{1})] $\overline{m}_{l}>0$ for all $l\in \mathcal{A}$, and

\item [(\rmnum{2})] $\overline{m}_{l}=0$ for all $l\in \{1, 2, \cdots, L\}-\mathcal{A}$.
\end{enumerate}
In the case that $\lambda_{l}^{\prime}=0$ (i.e., $\Lambda_{l}=\overline{m}_{l}\cdot \lambda_{max}^{(l)}$), due to (\ref{constraint-4}) alone, we derive that $y_{l}=0$. Similarly, we denote by $\mathcal{B}$ a subset of $\{1, 2, \cdots, L\}$ such that
\begin{enumerate}
\item [(\rmnum{1})] $\lambda_{l}^{\prime}>0$ for all $l\in \mathcal{B}$, and

\item [(\rmnum{2})] $\lambda_{l}^{\prime}=0$ for all $l\in \{1, 2, \cdots, L\}-\mathcal{B}$.
\end{enumerate}

\begin{observation}\label{transformation-3}
There exists a surjection from $\big( \{ x_{l} \}_{l\in\mathcal{A}},\,$ $\{ y_{l} \}_{l\in\mathcal{B}} \big)$ to $\left( \{ x_{l} \}_{l=1}^{L},\, \{ y_{l} \}_{l=1}^{L} \right)$ as follows:
\begin{equation}\label{equation-3}
\begin{split}
& x_{l} = x_{l}\; \text{for all}\, l\in\mathcal{A},\; \text{and},\; y_{l} = y_{l}\; \text{for all}\, l\in\mathcal{B}\\
& x_{l} = 0\; \text{for all}\, l\in\{1, 2, \cdots, L\} - \mathcal{A}\\
& y_{l} = 0\; \text{for all}\, l\in \{1, 2, \cdots, L\} - \mathcal{B}
\end{split}
\end{equation}
where $\left( \{ x_{l} \}_{l=1}^{L},\, \{ y_{l} \}_{l=1}^{L} \right)$ is subject to \eqref{constraint-4} and \eqref{constraint-2-1}; and $( \{ x_{l} \}_{l\in\mathcal{A}},\,$  $\{ y_{l} \}_{l\in\mathcal{B}} )$ is subject to
\begin{equation}\label{constraint-4-2}
\begin{split}
&\hskip-2mm 0\leq x_{l}\leq \overline{m}_{l},\; \text{for all}\, l\in\mathcal{A},\\
&\hskip-2mm y_{l}\in \{0, 1\} \, \text{if} \, x_{l}=\overline{m}_{l}, \text{and} \, y_{l}=0 \, \text{otherwise},\; \text{for all}\, l\in\mathcal{B},
\end{split}
\end{equation}
and,
\begin{align}
 \sum\nolimits_{l\in\mathcal{A}}{x_{l}}  +  \sum\nolimits_{l\in\mathcal{B}}{y_{l}} \leq m. \label{constraint-2-2}
\end{align}
\end{observation}

\begin{problem}\label{prob44}
Problem~\ref{prob22} is rewritten as:
\begin{center}
maximize\enskip $\sum\nolimits_{l\in\mathcal{A}}{x_{l}\cdot (\theta_{l}\cdot \lambda_{max}^{(l)})} + \sum\nolimits_{l\in\mathcal{B}}{y_{l}\cdot (\theta_{l}\cdot \lambda_{l}^{\prime})}$,
\end{center}
subject to constraints (\ref{constraint-4-2}) and (\ref{constraint-2-2}).
\end{problem}

Due to Observations~\ref{lemma-reduction-AA} and~\ref{transformation-3}, we conclude that
\begin{lemma}[Equivalence] \label{lemma-relationship-4}
An optimal solution to Problem~\ref{prob44} corresponds to an optimal solution to Problem~\ref{prob22} where these two optimal solutions satisfy (\ref{equation-3}).
\end{lemma}


Now, we proceed to reduce Problem~\ref{prob44} to another problem, whose optimal solution is observable.

\begin{observation}\label{transformation-2}
Let the decision variables $x_{l}$ and $y_{l}$ satisfy the constraints \eqref{constraint-4-2}, and \eqref{constraint-2-2}; let $x_{l,1}, \cdots, x_{l, \overline{m}_{l}}$ and $y_{l}$ satisfy
\begin{equation}\label{constraint-bit}
x_{l,1}, \cdots, x_{l, \overline{m}_{l}}\in \{0, 1\}\; \text{for all}\; l\in\mathcal{A}.
\end{equation}
\begin{equation}\label{constraint-3-2}
\begin{split}
 y_{l}\in \{0, 1\} \; \text{if} \; \sum\nolimits_{j=1}^{\overline{m}_{l}}{x_{l,j}}=\overline{m}_{l},  \, \text{and}, y_{l}=0 \, \text{otherwise}, & \\   \text{for all}\; l\in\mathcal{B}. &
\end{split}
\end{equation}
\begin{align}\label{constraint-3-1}
 \sum\nolimits_{l\in\mathcal{B}}{y_{l}} + \sum\nolimits_{l\in\mathcal{A}}{\sum\nolimits_{j=1}^{\overline{m}_{l}}{x_{l,j}} }\leq m.
\end{align}
There exists a surjection from $(x_{l,1}, \cdots, x_{l, \overline{m}_{l}}, y_{l})$ to $(x_{l}, y_{l})$ as follows:
\begin{equation}\label{relationship-3}
x_{l} = \sum\nolimits_{j=1}^{\overline{m}_{l}}{x_{l,j}},\; \text{and},\; y_{l}=y_{l}.
\end{equation}
\end{observation}

\begin{problem}\label{prob33}
Problem~\ref{prob44} is rewritten as:
\begin{equation}\label{eq:eqobjfun}
maximize\enskip  \sum\limits_{l\in\mathcal{A}}{ \sum\limits_{j=1}^{\overline{m}_{l}}{x_{l,j}\cdot (\theta_{l}\cdot \lambda_{max}^{(l)})} } + \sum\limits_{l\in\mathcal{B}}{ y_{l}\cdot (\theta_{l}\cdot \lambda_{l}^{\prime})},
\end{equation}
with $x_{l,j}$ and $y_{1}, \cdots y_{L}$ as decision variables and subject to constraints \eqref{constraint-bit}, \eqref{constraint-3-1} and \eqref{constraint-3-2}.
\end{problem}


Due to Observations~\ref{lemma-reduction-AA} and~\ref{transformation-2}, we conclude that
\begin{lemma}[Equivalence] \label{lemma-relationship-3}
An optimal solution to Problem~\ref{prob33} corresponds to an optimal solution to Problem~\ref{prob44} where these two optimal solutions satisfy (\ref{relationship-3}).
\end{lemma}

In the following, we solve Problem~\ref{prob33}.
We only need to determine the variables $x_{l,1}, \cdots, x_{l, \overline{m}_{l}}$ for all $l\in\mathcal{A}$ and the variable $y_{l}$ for all $l\in\mathcal{B}$.  For the coefficient of each variable $x_{l,j}$, we denote it by $\zeta_{l,j}^{(1)}$; for the coefficient of $y_{l}$, denote it by $\zeta_{l}^{(2)}$; here, we have  $\zeta_{l,j}^{(1)}> \zeta_{l}^{(2)}$ since $(\theta_{l}\cdot \lambda_{max}^{(l)}) > (\theta_{l}\cdot \lambda_{l}^{\prime})$. The objective function in Problem (\ref{prob33}) is equivalent to the following one:
\begin{equation}\label{eq:eqobjfun-1}
\sum\nolimits_{l\in\mathcal{A}}{ \sum\nolimits_{j=1}^{\overline{m}_{l}}{x_{l,j}\cdot \zeta_{l,j}^{(1)}} }  +  \sum\nolimits_{l\in\mathcal{B}}{ y_{l}\cdot \zeta_{l}^{(2)}}
\end{equation}

\begin{theorem}\label{theorem-capacity-planning-2}
Algorithm~\ref{algo-prob33-solution} produces an optimal solution to Problem~\ref{prob33}.
\end{theorem}

\begin{theorem}\label{theorem-capacity-planning-1}
Algorithm~\ref{algo-prob-solution} produces an optimal solution to Problem~\ref{prob1}.
\end{theorem}


\begin{algorithm}[t]
\SetKwInOut{Begin}{Begin}
\SetKwInOut{Input}{Input}
\SetKwInOut{Output}{Output}




select $m$ variables from $ x_{l,1}, \cdots, x_{l,\overline{m}_{l}}$ for all $l\in\mathcal{A}$ and $y_{l}$ for all $l\in\mathcal{B}$ such that the chosen variables are of the largest coefficients among all the variables in the objective function \eqref{eq:eqobjfun} or \eqref{eq:eqobjfun-1}, denoting these chosen variables by $\mathcal{D}$\;

set the values of the chosen variables to 1, and the values of the other variables to 0\;







\caption{An optimal solution to Problem~\ref{prob33}\label{algo-prob33-solution}}
\end{algorithm}

\begin{algorithm}[t]
\SetKwInOut{Begin}{Begin}
\SetKwInOut{Input}{Input}
\SetKwInOut{Output}{Output}


call Algorithm~\ref{algo-prob33-solution} to produce an optimal solution to Problem~\ref{prob33}\;

\tcc{\footnotesize{derive an optimal solution to Problem~\ref{prob44} by Lemma~\ref{lemma-relationship-3}}}

$x_{l} \leftarrow \sum_{j=1}^{\overline{m}_{l}}{x_{l,j}}$\, for all\, $l\in\mathcal{A}$\;

the value of $y_{l}$ in Problem~\ref{prob44} is equivalent to the value of $y_{l}$ in Problem~\ref{prob33}\;

\tcc{\footnotesize{derive an optimal solution to Problem~\ref{prob22} by Lemma~\ref{lemma-relationship-4}}}

$x_{l} \leftarrow 0$\, for all\, $l\in \{1, 2, \cdots, L\} - \mathcal{A}$\;

$y_{l} \leftarrow 0$\, for all\, $l\in \{1, 2, \cdots, L\} - \mathcal{B}$\;

\tcc{\footnotesize{derive an optimal solution to Problem~\ref{prob1} by Lemma~\ref{lemma-relationship-2}}}

$m_{l}\leftarrow x_{l}+y_{l}$\;

$\Lambda_{l}^{+} \leftarrow \lambda_{max}^{(l)}\cdot x_{l} + \lambda_{l}^{\prime}\cdot y_{l}$\;







\caption{An optimal solution to Problem~\ref{prob33}\label{algo-prob-solution}}
\end{algorithm}


\subsection{Prices of Different SLAs}
\label{sec.optimal-prices}

As elaborated at the end of Section~\ref{sec.prob-2}, we now solve (\ref{equa-pp-1}) optimally where $\mathcal{S}$ $=$ $\{l_{1}, \cdots, l_{L^{\prime}}\}$. The problem (\ref{equa-pp-1}) is as follows: when each of the $l_{1}$-th, $\cdots$, $l_{L^{\prime}}$-th SLAs will be chosen by some users and no users chooses the other SLAs, what is the optimal prices of SLAs? In the following, without loss of generality, we consider the case that each of the $L$ SLAs will be chosen by some users where $L^{\prime}=L$ and propose an algorithm to determine the optimal SLA prices $\theta_{1}^{*}$, $\theta_{2}^{*}$, $\cdots$, $\theta_{L}^{*}$. Such an algorithm is also applicable to the case that only a subset of SLAs will be chosen by users; here, the prices of the other SLAs are set to $\infty$ trivially.
We also assume that $\alpha_{1}$ $>$ $\alpha_{2}>\cdots>\alpha_{K}$ for the coefficients appearing in Definition~\eqref{def-utility}.
As a result, the utility functions of all users satisfy the following relation:
\begin{align}\label{equa-monotonicity}
\mathcal{U}_{1}(\varphi) > \mathcal{U}_{2}(\varphi) > \cdots > \mathcal{U}_{K}(\varphi),\; \text{for any time} \; \varphi\in\mathcal{R}.
\end{align}

In the following, we identify some structural property on users' choices of SLAs under an arbitrary setting of the SLA prices; then, we consider the case that the particular structure in an optimal solution is known in advance and derive the corresponding SLA optimal prices. Next, we search all the feasible structures for the structure under which the CSP achieves the highest revenue; the optimal prices under such a structure will be the optimal SLA prices of this paper.

\begin{lemma}\label{permutation}
Under an arbitrary price vector $\boldsymbol{\theta}\in \Theta_{\mathcal{S}}$ where $\mathcal{S}=\{ 1, \cdots, L \}$, users behave as follows:
there exist integers $1\leq i_{1} < \cdots < i_{L} \leq K$ such that, for all $l\in [1, L]$, the $l$-th SLA is to be chosen by users $i_{l-1}+1$, $\cdots$, $i_{l}$, where $i_{0} = 0$.
\end{lemma}
\begin{proof}
Let $1\leq j_{1} < j_{2} < L$, and $i_{1}^{\prime}$ and $i_{2}^{\prime}$ denote two users who will respectively choose the $j_{1}$-th and $j_{2}$-th SLAs; in this case, we could observe that $i_{1}^{\prime}<i_{2}^{\prime}$. We prove this observation by contradiction. Otherwise, we have $i_{1}^{\prime} > i_{2}^{\prime}$; the choices of users $i_{1}^{\prime}$ and $i_{2}^{\prime}$ alone imply the following relations on their surplus under the $j_{1}$-th and $j_{2}$-th SLAs: (\rmnum{1}) $\mathcal{U}_{i_{1}^{\prime}}(\varphi_{j_{1}}) - \theta_{j_{1}} \geq \mathcal{U}_{i_{1}^{\prime}}(\varphi_{j_{2}}) - \theta_{j_{2}}$, and (\rmnum{2}) $\mathcal{U}_{i_{2}^{\prime}}(\varphi_{j_{1}}) - \theta_{j_{1}} \leq \mathcal{U}_{i_{2}^{\prime}}(\varphi_{j_{2}}) - \theta_{j_{2}}$.
Multiplying the second inequality by -1 and adding it to the first inequality, we derive that
\begin{equation}\label{equa-contradiction-1}
\mathcal{U}_{i_{1}^{\prime}}(\varphi_{j_{1}}) - \mathcal{U}_{i_{1}^{\prime}}(\varphi_{j_{2}}) \geq \mathcal{U}_{i_{2}^{\prime}}(\varphi_{j_{1}}) - \mathcal{U}_{i_{2}^{\prime}}(\varphi_{j_{2}}).
\end{equation}
By Definition~\ref{def-utility}, the utility function satisfies Assumption~\ref{property-1}; by its 1st point, we have $\mathcal{U}^{\prime}(\varphi_{j_{1}}) > \mathcal{U}^{\prime}(\varphi_{j_{2}})>0$. Further, $\mathcal{U}_{i_{1}^{\prime}}(\varphi_{j_{1}}) - \mathcal{U}_{i_{2}^{\prime}}(\varphi_{j_{1}}) = (\alpha_{i_{1}^{\prime}}-\alpha_{i_{2}^{\prime}})\cdot \mathcal{U}^{\prime}(\varphi_{j_{1}}) < (\alpha_{i_{1}^{\prime}}-\alpha_{i_{2}^{\prime}})\cdot \mathcal{U}^{\prime}(\varphi_{j_{2}}) = \mathcal{U}_{i_{1}^{\prime}}(\varphi_{j_{2}}) - \mathcal{U}_{i_{2}^{\prime}}(\varphi_{j_{2}})$ where $\alpha_{i_{1}^{\prime}}-\alpha_{i_{2}^{\prime}}<0$ since $i_{1}^{\prime}>i_{2}^{\prime}$, which contradicts \eqref{equa-contradiction-1}. Now, we have completed proving the above observation.

Each SLA will be chosen by a subset of users. Now, consider the tagged users in the order of 1, 2, $\cdots$, $L$; let $i_{l}$ denote the last user who will choose the $l$-th SLA. Facing the $L$ SLAs, a user will finally choose some SLA if there exists a SLA under which its surplus is $>0$. Under the above observation, users $i_{l-1}+1$, $\cdots$, $i_{l}$ will only consider choosing the $l$-th SLA. The surplus of user $i_{l}$ under the $l$-th SLA is $>0$; so is the surplus of users $i_{l-1}+1$, $\cdots$, $i_{l}-1$ if $i_{l-1}+1$ $<$ $i_{l}$, by (\ref{equa-monotonicity}). Hence, all these users will choose the $l$-th SLA and Lemma~\ref{permutation} holds.
\end{proof}

A CSP's capacity is limited and possibly only a subset of users will get served; here, $i_{L} \leq K$. We denote by $\theta_{1}^{*}$, $\cdots$, $\theta_{L}^{*}$ an optimal solution for the SLA prices. Under such prices, we assume by Lemma~\ref{permutation} that the $l$-th SLA will be chosen by a subset of users $\mathcal{C}_{l}^{*}$ $=$ $\{i_{l-1}^{*}+1, \cdots, i_{l}^{*}\}$, which determines the total job arrival rate at the $l$-th processing unit, i.e., $\Lambda_{l}$ $=$ $\sum_{j=i_{l-1}^{*}+1}^{i_{l}^{*}}{h_{j}}$ as defined by (\ref{equa-rate}). With the SLA prices and job arrival rates, we could derive the maximum revenue (\ref{equa-max}) of a CSP by Algorithm~\ref{algo-prob-solution} that will determine the rate $\Lambda_{l}^{+}$ of accepting jobs at every unit.


In the following, we analyze what are the optimal prices of SLAs to maximize the revenue (\ref{equa-max}) in the case that we know the users' choices of an optimal solution where the $l$-th SLA is chosen by users $\mathcal{C}_{l}^{*}$. In our subsequent analysis, if $|\mathcal{C}_{l}^{*}|=1$ where $\mathcal{C}_{l}^{*} = \{ i_{l}^{*} \}$, we assume that the maximum revenue is achieved when the $l$-th SLA are chosen by $\mathcal{C}_{l}^{*}$; however, if user $i_{l}^{*}$ changes its choice to the ($l+1$)-th SLA, a lower revenue (\ref{equa-max}) will be achieved by Algorithm~\ref{algo-prob-solution}. Such a particular assumption does not affect us to derive an optimal solution for the SLA prices under which the revenue (\ref{equa-max}) is maximized. The reason is as follows. Suppose that there exists some $l\in [1, \, L-1]$ such that the maximum revenue is generated no mater whether the user $i_{l}^{*}$ chooses the $l$-th or ($l+1$)-th SLA. For all such $l$, we consider a new case on users' choices: the user $i_{l}^{*}$ changes its choice to the ($l+1$)-th SLA and the choices of the other users do not change, after which no users chooses the $l$-th SLA. In this new case, the maximum revenue could still be achieved and it fulfills the above assumption although possibly not all the $L$ SLAs might be chosen by users where $|\mathcal{S}|\leq L$; only if we could give the corresponding optimal SLA prices where the CSP achieves the maximum revenue under such users' choices, we will have obtained an optimal solution for the SLA prices among all possible optimal solutions.

\begin{definition}\label{def-optip}
For all $l\in [1, L]$, let $\mathcal{U}_{i_{l}^{*}}^{-}$ $=$ $\mathcal{U}_{i_{l}^{*}}(\varphi_{l})$ $-$ $\mathcal{U}_{i_{l}^{*}}(\varphi_{l+1})$. We define parameter $\theta_{l}^{\prime}$ to be such that,
\begin{enumerate}
\item [(\rmnum{1})] $\theta_{L}^{\prime}$ is a value smaller than but close to $\mathcal{U}_{i_{L}^{*}}(\varphi_{L})$, i.e.,  $\theta_{L}^{\prime}=\mathcal{U}_{i_{L}^{*}}(\varphi_{L})-\epsilon$, where $\epsilon$ is a small enough positive real number;
\item [(\rmnum{2})] for all $l\in [1,\, L-1]$, $\theta_{l}^{\prime}$  is the maximum possible $\theta_{l}$ that satisfies $\theta_{l} < \mathcal{U}_{i_{l}^{*}}^{-} + \theta_{l+1}^{\prime}$, i.e., $\theta_{l}^{\prime}=\mathcal{U}_{i_{l}^{*}}^{-}+\theta_{l+1}^{\prime}-\epsilon$.
\end{enumerate}
\end{definition}


\begin{theorem}\label{theo-sequence}
In an optimal solution where a CSP achieves the maximum revenue (\ref{equa-max}), we assume for all $l\in [1, L]$ that the $l$-th SLA will be chosen by users $\mathcal{C}_{l}^{*} = \{ i_{l-1}^{*}+1, \cdots, i_{l}^{*} \}$. Given the knowledge of $\mathcal{C}_{1}^{*}$, $\cdots$, $\mathcal{C}_{L}^{*}$, the corresponding optimal price $\theta_{l}^{*}$ of the $l$-th price is $\theta_{l}^{\prime}$ in Definition~\ref{def-optip}.
\end{theorem}
\begin{proof}
We first prove that, when the SLA prices are $\theta_{1}^{\prime}$, $\cdots$, $\theta_{L}^{\prime}$ in Definition~\ref{def-optip}, we could derive for all $l\in [1, L]$ that the $l$-th SLA will be chosen by users $i_{l-1}^{*}+1$, $\cdots$, $i_{l}^{*}$; this will show that $\theta_{1}^{\prime}$, $\cdots$, $\theta_{L}^{\prime}$ are a feasible solution. For all $i\in [i_{l-1}^{*}+1,\, i_{l}^{*}]$, the surplus of user $i$ respectively under the $l^{\prime}$-th and ($l^{\prime}+1$)-th SLA is denoted by $g_{i}^{(l^{\prime})} = \mathcal{U}_{i}(\varphi_{l^{\prime}})-\theta_{l^{\prime}}^{\prime}$, and $g_{i}^{(l^{\prime}+1)} = \mathcal{U}_{i}(\varphi_{l^{\prime}+1})-\theta_{l^{\prime}+1}^{\prime}$ where $l^{\prime}\in [1, L-1]$; then, recall Definition~\ref{def-utility} where Assumption~\ref{property-1} is satisfied and we have
\begin{center}
$g_{i}^{(l^{\prime})} - g_{i}^{(l^{\prime}+1)} = (\alpha_{i}-\alpha_{i_{l^{\prime}}^{*}})\cdot (\mathcal{U}^{\prime}(\varphi_{l^{\prime}}) - \mathcal{U}^{\prime}(\varphi_{l^{\prime}+1})) + \epsilon$.
\end{center}
By Assumption~\ref{property-1}, $\mathcal{U}^{\prime}(\varphi)$ decreases as $\varphi$ increases, and thus
\begin{enumerate}
\item [(\rmnum{1})] if $l=L$, we have $\alpha_{i}-\alpha_{i_{l^{\prime}}^{*}} < 0$ and further have $g_{i}^{(l^{\prime})} < g_{i}^{(l^{\prime}+1)}$, that is, $g_{i}^{(L)}>g_{i}^{(L-1)}>\cdots>g_{i}^{(1)}$;
\end{enumerate}
similarly, we also have that
\begin{enumerate}
\item [(\rmnum{2})] if $l\in [2, \, L-1]$, $g_{i}^{(1)}< \cdots  < g_{i}^{(l)} > \cdots > g_{i}^{(L)}$.

\item [(\rmnum{3})] if $l=1$, we have $g_{i}^{(1)}>g_{i}^{(2)}>\cdots>g_{i}^{(L)}$.
\end{enumerate}
Hence, we have for all $i \in [i_{l-1}^{*}+1, i_{l}^{*}]$ that user $i$ achieves the highest surplus only under the $l$-th SLA, and it will choose the $l$-th SLA.

As shown by (\ref{equa-max}), no matter what the optimal rate $\Lambda_{l}^{+}$ of accepting jobs at every unit $l$ is, larger SLA prices could always lead to larger revenue; hence, given the users' choices, the maximum possible SLA prices under which users will make such choices are also the optimal SLA prices $\theta_{l}^{*}$, $\cdots$, $\theta_{L}^{*}$. Above, we have proved that $\theta_{1}^{\prime}, \cdots, \theta_{L}^{\prime}$ are feasible SLA prices such that for all $l\in [1, L]$ the $l$-th SLA is chosen by users $\mathcal{C}_{l}^{*}$. Now, we prove that they are also the maximum possible, i.e., $\theta_{l}^{*} = \theta_{l}^{\prime}$ for all $l\in [1, L]$.

In terms of the $L$-th SLA, the utilities of users satisfy $\mathcal{U}_{i_{L-1}^{*}+1}(\varphi_{L})$ $>$ $\cdots$ $>$ $\mathcal{U}_{i_{L}^{*}}(\varphi_{L})$ and, in order to ensure that users $i_{L-1}^{*}+1,\cdots, i_{L}^{*}$ will choose the $L$-th SLA, a necessary condition is that the surplus of each user should be $>0$; hence, we have $\mathcal{U}_{i_{L}^{*}}(\varphi_{L})-\theta_{l} > 0$, and the maximum possible price of the $L$-th SLA is $\theta_{L}^{\prime}$. For all $l\in [1, L-1]$, in terms of the $l$-th SLA, the surplus of user $i_{l}^{*}$ achieves the maximum surplus under the $l$-th SLA, and we have
\begin{align}
& \mathcal{U}_{i_{l}^{*}}(\varphi_{l}) - \theta_{l} \geq \mathcal{U}_{i_{l}^{*}}(\varphi_{l+1}) - \theta_{l+1}\label{price-1}.
\end{align}
Given the value of $\theta_{l+1}$, the maximum possible price $\theta_{l}^{*}$ of the $l$-th SLA is either (\rmnum{1}) $\mathcal{U}_{i_{l}^{*}}^{-}+\theta_{l+1}$ when the equal sign of (\ref{price-1}) holds, or (\rmnum{2}) $\mathcal{U}_{i_{l}^{*}}^{-}+\theta_{l+1}-\epsilon$ when the left side of (\ref{price-1}) is strictly greater than its left side.

Suppose that there exists some $l\in [1,\, L-1]$ such that $\theta_{l}^{*}$ is $\mathcal{U}_{i_{l}^{*}}^{-}+\theta_{l+1}$.
Under such prices, we have that user $i_{l}^{*}$ achieves the same surplus under the $l$-th and ($l+1$)-th SLAs, i.e., the equal sign in (\ref{price-1}) holds; in other words, when user $i_{l}^{*}$ makes choice, it will randomly choose the $l$-th or ($l+1$)-th SLA, as indicated in Section~\ref{sec.user-choice}. Since we are considering the users' choices under which a CSP could achieve the maximum revenue once the corresponding optimal SLA prices are also given, the existence of the equal-sign case of (\ref{price-1}) means that, with the current price setting, a CSP could achieve the same maximum revenue no matter which of the $l$-th and ($l+1$)-th SLAs is chosen by user $i_{l}^{*}$ where $|C_{l}^{*}|\geq 2$; otherwise, the left of (\ref{price-1}) should be $>$ its right in order to ensure that user $i_{l}^{*}$ will choose the $l$-th SLA definitely.

However, for all such $l\in [1,\, L-1]$, we can consider them in the decreasing order, and increase the price of the $l$-th SLA to $\mathcal{U}_{i_{l}^{*}-1}(\varphi_{l})$ $-$ $\mathcal{U}_{i_{l}^{*}-1}(\varphi_{l+1}) +\theta_{l+1} - \epsilon$; then, user $i_{l}^{*}$ will change its choice to the ($l+1$)-th SLA, the last user who chooses the $l$-th SLA will become user $i_{l}^{*}-1$, and the choices of the other users do not change. In particular, like our analysis results for $\theta_{1}^{\prime}, \cdots, \theta_{L}^{\prime}$ at the beginning, we can assume after the price increment that the $l$-th SLA will be chosen by users $\tilde{i}_{l-1}^{*}+1$, $\cdots$, $\tilde{i}_{l}^{*}$ where $\tilde{i}_{L}^{*}$ is still $i_{L}^{*}$ and $\tilde{i}_{0}^{*}$ is set to zero; here, the price $\theta_{l}$ of the $l$-th SLA is still of the same form as $\theta_{l}^{\prime}$ in Definition~\ref{def-optip}, that is, (\rmnum{1}) $\theta_{L}=\mathcal{U}_{\tilde{i}_{L}^{*}}(\varphi_{L})-\epsilon$, and (\rmnum{2}) $\theta_{l} = ( \mathcal{U}_{\tilde{i}_{l}^{*}}(\varphi_{l}) - \mathcal{U}_{\tilde{i}_{l}^{*}}(\varphi_{l+1}) ) + \theta_{l+1}-\epsilon$ for all $l\in [1, L-1]$.

In the case without the price increment, even if user $i_{l}^{*}$ changes its choice from the $l$-th SLA to the ($l+1$)-th SLA, the CSP can still achieve the same maximum revenue under the prices $\theta_{1}^{*}$, $\cdots$, $\theta_{L}^{*}$; however, as shown by the definition of the revenue (\ref{equa-max}), we have after the price increment that the CSP will achieve a higher revenue, which contradicts the fact that under the users' choices defined by $i_{1}^{*},\cdots, i_{L}^{*}$ a CSP could achieve the maximum revenue.
\end{proof}

By Theorem~\ref{theo-sequence}, the $L$-th SLA's optimal price is
\begin{equation}\label{equa-price-L}
\theta_{L}^{*} = \mathcal{U}_{i_{L}^{*}}(\varphi_{L}) - \epsilon;
\end{equation}
the ($L-1$)-th SLA's optimal price could be expressed as
\begin{center}
$\theta_{L-1}^{*}=\mathcal{U}_{i_{L-1}^{*}}(\varphi_{L-1}) - ( \mathcal{U}_{i_{L-1}^{*}}(\varphi_{L}) - \mathcal{U}_{i_{L}^{*}}(\varphi_{L}) ) - 2\epsilon$,
\end{center}
i.e., a value smaller than but close enough to the utility of user $i_{L-1}^{*}$ under the ($L-1$)-th SLA (i.e., when the waiting time is $\varphi_{L-1}$) minus the difference of the utilities of users $i_{L-1}^{*}$ and $i_{L}^{*}$ under the $L$-th SLA. Generalizing this, for all $l\in [1,\, L-1]$, the $l$-th SLA's optimal price is
\begin{equation}\label{equa-price-ll}
\theta_{l}^{*} = \mathcal{U}_{i_{l}^{*}}(\varphi_{l}) - \sum\nolimits_{j=l}^{L-1}{D_{j}^{*}},
\end{equation}
where $D_{j}^{*}$ denotes the difference of the utilities of users $i_{j}^{*}$ and $i_{j+1}^{*}$ under the ($j+1$)-th SLA, i.e., $D_{j}^{*} = \mathcal{U}_{i_{j}^{*}}(\varphi_{j+1}) - \mathcal{U}_{i_{j+1}^{*}}(\varphi_{j+1})$.

In fact, Theorem~\ref{theo-sequence} and Lemma~\ref{permutation} are the main results of this subsection and we illustrate this in Figure~\ref{Fig.10}; here, there exists a sequence $1\leq i_{1}^{*}<\cdots<i_{L}^{*}\leq K$ such that the optimal prices are determined by (\ref{equa-price-L}) and (\ref{equa-price-ll}) where the utility $\mathcal{U}_{i_{l}^{*}}(\varphi_{l})$ of user $i_{l}^{*}$ under the $l$-th SLA is simply denoted by $\mathcal{U}_{i_{l}^{*}}$. Finally, we have that
\begin{theorem}\label{theorem-optimal-prices}
Algorithm~\ref{OptiPrices} returns the optimal prices $\boldsymbol{\theta^{*}}$ of different SLAs in polynomial time.
\end{theorem}
\begin{proof}
From Theorem~\ref{theo-sequence}, once we obtain the sequence $i_{1}^{*}$, $i_{2}^{*}$, $\cdots$, $i_{L}^{*}$ in the optimal solution, we could derive the optimal SLA prices, leading to the maximum revenue (\ref{equa-max}). The optimal sequence is among all the subsets of $\{1, 2, \cdots, K\}$ whose cardinality is $L$. Algorithm~\ref{OptiPrices} considers every element $\{i_{1},\, i_{2}\, \cdots,\, i_{L}\}$ of a $L$-combination of $\{1, 2, \cdots, K\}$ and computes the corresponding prices of SLAs. Next, it selects the element that generates the highest revenue and returns the SLA prices when using this element to replace the sequence in Theorem~\ref{theo-sequence} and to compute the SLA prices by Theorem~\ref{theo-sequence}; hence, the returned prices are the optimal SLA prices. The time complexity of checking each element of a $L$-combination of $\{1, 2, \cdots, K\}$ is $\binom{K}{L}=\frac{K!}{(K-L)!L!}$. A CSP would provide a finite number of SLAs and $L$ could be bounded by a constant. Hence, Algorithm~\ref{OptiPrices} has a time complexity polynomial in $K$.
\end{proof}

\begin{figure}
\begin{center}
  \includegraphics[width=3.3in]{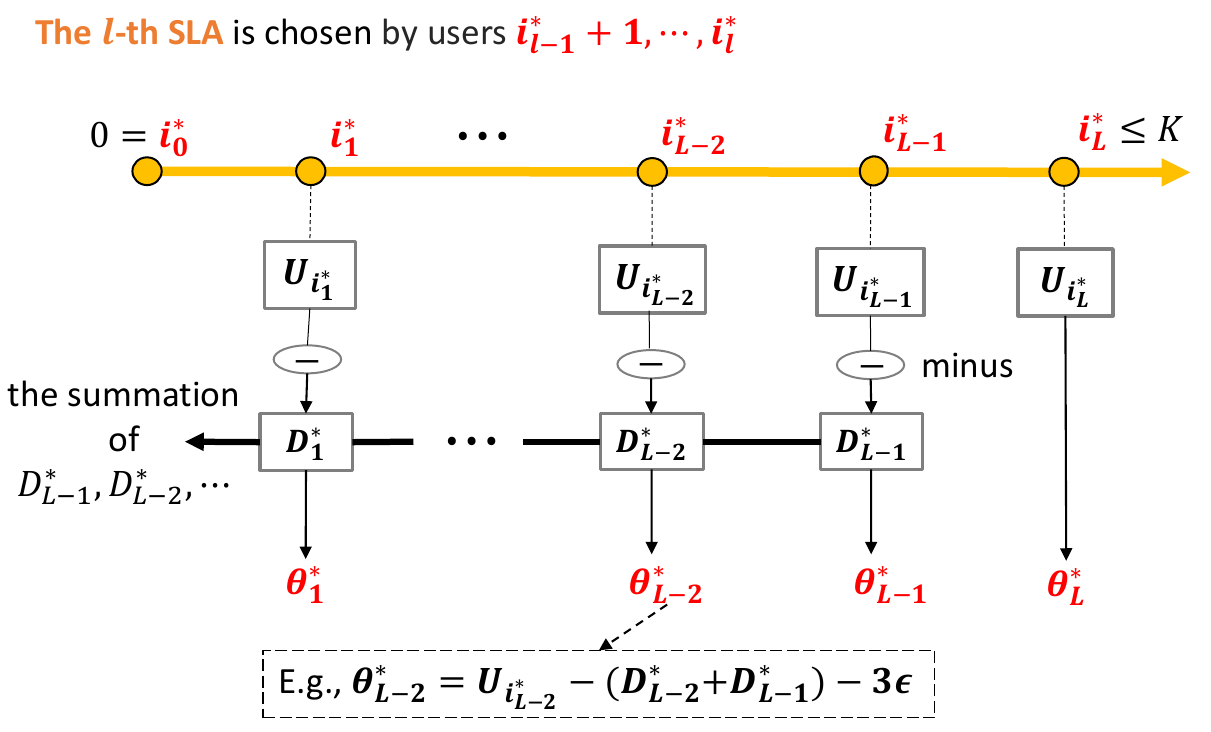}
  \caption{Structural Results in the Optimal Pricing (also see (\ref{equa-price-L}) and (\ref{equa-price-ll})).}\label{Fig.10}
\end{center}
\end{figure}

\begin{algorithm}[t]
\SetKwInOut{Begin}{Begin}
\SetKwInOut{Input}{Input}
    \setlength\itemsep{0.5em}

$\mathcal{A}^{\prime} \leftarrow \left\{ (i_{1}, \cdots, i_{L}) | 1 \leq i_{1} < \cdots < i_{L} \leq K \right\}$\tcp*{\footnotesize{$\mathcal{A}^{\prime}$ is a $L$-combination of $\{1, 2, \cdots, K\}$.}}

\vspace{0.1em}for every $(i_{1}, \cdots, i_{L})\in\mathcal{A}^{\prime}$, use it to replace the sequence $i_{1}^{*}, \cdots, i_{L}^{*}$ in Theorem~\ref{theo-sequence} and compute the corresponding SLA prices $( \theta_{1}^{\prime}, \cdots, \theta_{L}^{\prime} )$ by Theorem~\ref{theo-sequence}\;

\vspace{0.15em}$\mathcal{A} \leftarrow \left\{  \left(  \theta_{1}^{\prime}, \theta_{2}^{\prime}, \cdots, \theta_{L}^{\prime} \right)  | (i_{1}, \cdots, i_{L})\in \mathcal{A}^{\prime} \right\}$\;

\vspace{0.25em}$\boldsymbol{\theta^{*}} \leftarrow \arg\max\limits_{\boldsymbol{\theta}\in\mathcal{A}}\, \mathcal{V}\left( \boldsymbol{m(\theta)}, \boldsymbol{\Lambda^{+}(\theta)}, \boldsymbol{\theta} \right)$, defined in (\ref{fun-optimal-prices})\tcp*{\footnotesize{the optimal prices of SLAs are returned, achieving the highest revenue.}}

\caption{Optimal Prices of Different SLAs\label{OptiPrices}}
\end{algorithm}



\section{Performance Evaluation}
\label{sec.performance-evaluation}


In this section, we numerically evaluate the performance of the proposed QoS-differentiated pricing model. The key
performance metric is how much the unit revenue of the CSP improves, compared with the pure usage-based pricing model without QoS-differentiation.
In the usage-based pricing model, users are charged a fixed price whenever they utilize a server for a unit of time and it
corresponds to the first SLA in our model as explained in Section~\ref{sec_pricing}.

\subsection{Simulation Setup}

We assume a total number of $K=50$ users. Also, the job arrival rate $h_{i}$ of each user $i$ is set to $20$ jobs in the
time unit. Such parameters are kept fixed in order to evaluate the impact of other factors' affects onto the performance of the
QoS-differentiated pricing model. The performance of a QoS-differentiated pricing model and user's choices of SLAs are jointly
affected by the following \text{factors}: \textbf{(\rmnum{1})} the number $m$ of servers possessed by a CSP, \textbf{(\rmnum{2}} the user population's sensitivity to latency, \textbf{(\rmnum{3})} the weights of different users, namely the $\alpha_i$s in \eqref{def-utility}), \textbf{(\rmnum{4})} the degree to which the server utilization increases with the waiting time, and \textbf{(\rmnum{5})} the number $L$ of SLAs offered to users, and the chosen waiting times $\varphi_{1}, \cdots, \varphi_{L}$, stated in Section~\ref{sec.pricing-model}.
The utility function employed is the one described in Definition~\ref{def-utility} and further specified by \eqref{equa-uf1}.

Our performance evaluation is performed by exploring these five parameters in diverse cases.

\noindent\textbf{Factor 1.} We evaluate three cases where a CSP has $1000$, $2000$, and $4000$ servers, respectively.\\
\noindent\textbf{Factor 2.} We consider three cases where a user population has {\em high}, {\em medium}, and {\em low} sensitivity to latency; correspondingly, we set $\beta$ to $0.75$, $0.45$, and $0.25$ respectively where $\psi = 0.3$. The effect of $\beta$ on the user population's utility (or willingness to pay) is illustrated in Figure~\ref{Fig.17}, e.g., in the high case, the user's willingness to pay (or utility) (when the waiting time is $1.6$) decreases to about $0.25$ times the willingness to pay when the waiting time is 0.033.\\
\noindent\textbf{Factor 3.} We consider two cases where the $\alpha_i$s' distribution is {\em compact} and {\em loose} respectively. In the former case, for all $l\in [1,\, 50]$, the weight of the $i$-th user is set to $a^{N+1-i}$,  where $N=50$ and $a=1.028$; we observe that with this choice, for the assigned sensitivity to latency (i.e., the value of $\beta$), the ratio of the largest weight of users to the smallest one is 3.870. In the latter case, by choosing $a=1.05$, this ratio is rendered larger, e.g., $10.92$.\\
\noindent\textbf{Factor 4.} Given users' willingness to pay, the smaller the degree to which utilization increases with waiting time, the smaller the performance improvement of the QoS-differentiated pricing model over the usage-based pricing model. As illustrated in Figure~\ref{Fig.12}, we consider the worst case among them where the Random policy is used, which is also used in \cite{Zheng16a} for cloud brokerage services, and the job's runtime follows an exponential distribution, to show the viability of the pricing model of this paper.\\
\noindent\textbf{Factor 5.} We consider $L=6$ SLAs and artificially set $\varphi_{1}=0.033$, $\varphi_{2}=0.1$, $\varphi_{1}=0.2$, $\varphi_{1}=0.4$, $\varphi_{1}=0.8$, and $\varphi_{1}=1.6$; here, as the waiting time requirement changes from $\varphi_{l}$ to $\varphi_{l+1}$, the utilization has a remarkable increase, as illustrated in Figure~\ref{Fig.12}.

\begin{figure}
\begin{center}
  \includegraphics[width=1.99in]{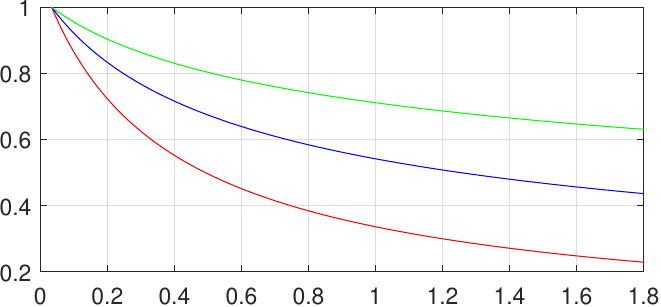}
  \caption{Three curves of $\mathcal{U}(\varphi, \alpha)/\mathcal{U}(0.033, \alpha)$ with $\beta=$ 0.75, 0.55, 0.2 from top to bottom, implying how much the utility of a population of users is decreased by as the waiting time increases.}\label{Fig.17}
\end{center}
\end{figure}

\vspace{0.28em}\noindent\textbf{Performance Metrics.} In this paper, the main performance metric is the revenue improvement of our model compared to usage-based pricing. In particular, the parameter $r_{i,j}^{(k)}$ is used to denote the relative gain of the QoS-differentiated scheme, i.e., the ratio of the total unit revenue gain of our model to the total unit revenue of the usage-based pricing model. The ratio is computed for various simulation configurations, $k$ indexing factor 1, i.e., the number of servers, $i$ indexing factor $2$, i.e., the user population's sensitivity to latency, and $j$ indexing factor 3, i.e., the weight distribution. In particular, $r_{1,j}^{(k)}$, $r_{2,j}^{(k)}$, and $r_{3,j}^{(k)}$ denote the cases $\beta=0.75,0.45,0.25$, respectively. Also, $r_{i,1}^{(k)}$ and $r_{i,2}^{(k)}$ identify the compact and loose cases for the weight distribution, respectively. Finally, $r_{i,j}^{(1)}$, $r_{i,j}^{(2)}$, and $r_{i,j}^{(3)}$ denote the parameter values when a CSP has $1000$, $2000$, and $4000$ servers, respectively. For example, $r_{1,1}^{(2)}$ denotes the ratio in the case $\beta = 0.75$, compact weight distribution, and $2000$ servers.


\begin{table*}[htbp]
  \begin{center}
  	\begin{threeparttable}
  	
		\begin{tabular}{| C{1.35cm} | C{0.7cm}  | C{1.7cm}  | C{1.8cm} | C{1.8cm} | C{1.8cm} | C{1.7cm} ||  C{2.6cm}   |}   
	
\hline

			            &   1st SLA   &  2nd SLA   &  3rd SLA   &  4th SLA   &  5th SLA  &  6th SLA   &   Revenue Improvement ($r_{i,j}^{(k)}$)\\ \hline

(loose, 1000, 0.55) &  $\times$   &   $\times$   &  ({\color{purple(html/css)} 2}, {\color{green(pigment)} 20.81}, {\color{frenchblue} 40}, {\color{red} 240})   &  ({\color{purple(html/css)} 10}, {\color{green(pigment)} 16.16}, {\color{frenchblue} 160}, {\color{red} 560})   &  ({\color{purple(html/css)} 13}, {\color{green(pigment)} 12.61}, {\color{frenchblue} 60}, {\color{red} 135})   &  ({\color{purple(html/css)} 15}, {\color{green(pigment)} 9.64}, {\color{frenchblue} 40}, {\color{red} 65})     &  258.5\%  \\ \hline

(loose, 2000, 0.55)&   $\times$ &  ({\color{purple(html/css)} 3}, {\color{green(pigment)} 20.03}, {\color{frenchblue} 60}, {\color{red} 661})	&({\color{purple(html/css)} 11}, {\color{green(pigment)} 16.7}, {\color{frenchblue} 160}, {\color{red} 960}) & ({\color{purple(html/css)} 14}, {\color{green(pigment)} 13.69}, {\color{frenchblue} 60}, {\color{red} 210})  &  ({\color{purple(html/css)} 17}, {\color{green(pigment)} 10.77}, {\color{frenchblue} 60}, {\color{red} 136})   &  ({\color{purple(html/css)} 18}, {\color{green(pigment)} 8.33}, {\color{frenchblue} 20}, {\color{red} 33})   &    138.7\%   \\ \hline


(loose, 2000, 0.2)	&     $\times$   & ({\color{purple(html/css)} 4}, {\color{green(pigment)} 17.46}, {\color{frenchblue} 80}, {\color{red} 880})  & ({\color{purple(html/css)} 11}, {\color{green(pigment)} 13.25}, {\color{frenchblue} 140}, {\color{red} 840})  & ({\color{purple(html/css)} 15}, {\color{green(pigment)} 9.63}, {\color{frenchblue} 80}, {\color{red} 280})   &    $\times$   &   $\times$    &   110.9\%  \\ \hline

(compact, 2000, 0.55)  &  $\times$  & $\times$	&   ({\color{purple(html/css)} 6}, {\color{green(pigment)} 6.733}, {\color{frenchblue} 120}, {\color{red} 720})    & ({\color{purple(html/css)} 19}, {\color{green(pigment)} 5.256}, {\color{frenchblue} 260}, {\color{red} 910})  & ({\color{purple(html/css)} 25}, {\color{green(pigment)} 4.093}, {\color{frenchblue} 120}, {\color{red} 271})   &  ({\color{purple(html/css)} 28}, {\color{green(pigment)} 3.141}, {\color{frenchblue} 60}, {\color{red} 98})   &   234.7\%   \\ \hline

		\end{tabular}
		\caption{The Optimal Solutions under Different Cases: the {\color{red} red}, {\color{green(pigment)} green} and {\color{frenchblue} blue} numbers denotes the number of assigned servers, the SLA price, and the optimal rate of accepting jobs, respectively.}
		\label{table-PC1}
	\end{threeparttable}
  \end{center}
\end{table*}

\begin{table}[htbp]
  \begin{center}
  	\begin{threeparttable}
  	

		\begin{tabular}{|  C{1.62cm}  | C{1.62cm} | C{1.62cm} | C{1.65cm}  |}   
	
\hline

(loose, 1000, 0.55)   &  (loose, 2000, 0.55)   &    (loose, 2000, 0.2)   &  (compact, 2000, 0.55)  \\ \hline

({\color{purple(html/css)} 2}, {\color{green(pigment)} 39.81}, {\color{frenchblue} 31.95}, {\color{red} 1000})   &	
({\color{purple(html/css)} 4}, {\color{green(pigment)} 36.11}, {\color{frenchblue} 63.89}, {\color{red} 2000})   &
({\color{purple(html/css)} 4}, {\color{green(pigment)} 29.84}, {\color{frenchblue} 63.89}, {\color{red} 2000})   &
({\color{purple(html/css)} 4}, {\color{green(pigment)} 13.35}, {\color{frenchblue} 63.89}, {\color{red} 2000})   \\ \hline

		\end{tabular}
		\caption{The Optimal Solutions for the On-demand Pricing Model under Different Cases.}
		\label{table-PC2}
	\end{threeparttable}
  \end{center}
\end{table}

\begin{figure}
\begin{center}
  \includegraphics[width=3.28in]{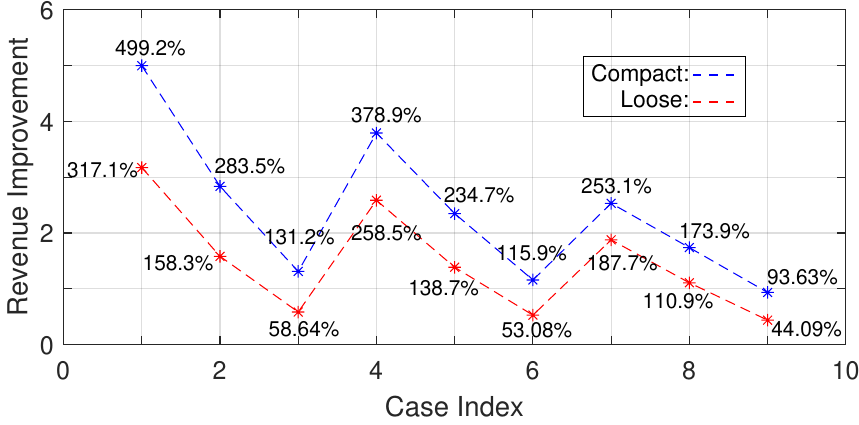}
  \caption{Revenue Improvement: the red curve corresponds to the loose case and the blue one to the tight case, respectively.}\label{Fig.10}
\end{center}
\end{figure}

\subsection{Results}

\subsubsection{Overview}
The main numerical results are illustrated in Figure~\ref{Fig.10} where the values of $r_{i,1}^{(k)}$ and $r_{i,2}^{(k)}$ correspond to the blue and red ones, respectively. Also, for the sake of representation, the cases are indexed using a single parameter $3\cdot (i-1)+k$. The revenue improvement of our model over the standard usage-based pricing model ranges from 499.2\% to 93.63\% in the compact case and from 317.1\% to 44.09\% in the loose case, depending on parameters $m$ and $\beta$. In all cases, the revenue improvements are either very large (e.g., up to 499.2\%) or large (e.g., 44.09\% in the worst case).
From Figure~\ref{Fig.10}, we could observe the following \textbf{three phenomena}:
\begin{enumerate}

\item [\textbf{(\rmnum{1})}] Under the same weight distribution and $\beta$, and, given the user population size, the revenue improvement decreases with the number of servers a CSP possesses.

\item [\textbf{(\rmnum{2})}] The smaller the value of $\beta$ is, the smaller the revenue improvement is.

\item [\textbf{(\rmnum{3})}] The more compact the weight distribution of users is, the larger the revenue improvement is.

\end{enumerate}

The first phenomenon is obtained by observing the values of $r_{i,j}^{(1)}$, $r_{i,j}^{(2)}$, and $r_{i,j}^{(3)}$ given the values of $i$ and $j$ (i.e., the values of Y coordinate when the values of X coordinate are $3\cdot (i-1)+1$, $3\cdot (i-1)+2$, and $3\cdot (i-1)+3$ in both the loose and compact cases where $i=1$, $2$, $3$). The second phenomenon is obtained by observing the values of $r_{1,j}^{(k)}$, $r_{2,j}^{(k)}$, and $r_{3,j}^{(k)}$ (i.e., the three values of Y coordinate in either the loose case or the compact case when the values of X coordinate are $k$, $3+k$, and $6+k$ where $k$ equals $1$, $2$, $3$). The third phenomenon is obtained by observing the values of $r_{i,1}^{(k)}$ and $r_{i,2}^{(k)}$ (i.e., the two values of Y coordinate respectively in the loose and compact cases when the value of X coordinate is $3\cdot (i-1)+k$ where $i$ and $k$ equal $1$, $2$, $3$).

As a summary of these three phenomena, we conclude that, when a CSP adopts the QoS-differentiated pricing model, a significant revenue improvement could be achieved especially in the cases where the population of users is less sensitive to latency, the weights of users are more compact, and the population of users is large compared with its processing capacity (i.e., the number of possessed servers).

\subsubsection{Case Analysis}

In the following, we take a deeper look at the concrete numerical results in some typical cases, which are given in Table~\ref{table-PC1}. Here, each item in the first column specifies the simulation setup, e.g., the item in the second row implies that the weight distribution is compact, $m=1000$, and $\beta=0.55$; thus, in Table~\ref{table-PC1}, the rows 2-4, 5, and 6 respectively give the simulation results in the loose case with $\beta=0.55$, the loose case with $\beta=0.2$ and $m=2000$, and the compact case with $\beta=0.55$ and $m=2000$. The last column of Table~\ref{table-PC1} gives the revenue improvements under these different cases. In Table~\ref{table-PC1}, in each row, all tuples $(y_{1}^{(l)}, y_{2}^{(l)}, y_{3}^{(l)}, y_{4}^{(l)})$ under the $l$-th SLA constitute an optimal solution in the case specified by the values of the first column; the cross denotes that this SLA will not be offered to users. Now, we explain the meaning of each tuple. The $l$-th SLA's price is specified by the greened $y_{2}^{(l)}$; at the $l$-th processing unit, the optimal rate of accepting jobs and the number of assigned servers are respectively the blued $y_{3}^{(l)}$ and the red $y_{4}^{(l)}$. In the case that the ($l-1$)-th SLA is offered to users, the indexes of users that will choose the $l$-th SLA are $y_{1}^{(l-1)}+1, \cdots, y_{1}^{(l)}$; in the case that $l=1$ or the ($l-1$)-th SLA is not offered, the corresponding indexes of users are $1, \cdots, y_{l}^{(l)}$.

When analyzing these concrete results to better understand the logic under which the above three phenomena occur, we need to keep the following points in mind. The unit revenue from a server in a unit of time is the product of the server's utilization and the price associated with it\emph{}. Given the $l$-th SLA offered to users, let $l^{\prime}$ denote the total number of offered SLAs after the $l$-th SLA, e.g., in the 2nd row of Table~\ref{table-PC1}, if $l=4$, $l^{\prime}=2$. As illustrated in Figure~\ref{Fig.10}, the price of the $l$-th SLA is
$\mathcal{U}_{y_{1}^{(l)}}(\varphi_{l}) - \sum_{j=l}^{l+l^{\prime}-1}{D_{j}^{*}} - (l^{\prime}+1)\cdot\epsilon$ if $l^{\prime} > 0$ and
$\mathcal{U}_{y_{1}^{(l)}}(\varphi_{l}) - \epsilon$ if $l^{\prime}=0$;
here, $\mathcal{U}_{y_{1}^{(l)}}(\varphi_{l})$ is the utility of the $y_{1}^{(l)}$-th user when the waiting time is $\varphi_{l}$, and $D_{j}^{*}$ is the difference of the utilities of users $y_{1}^{(j)}$ and $y_{1}^{(j+1)}$ when the waiting time is $\varphi_{j+1}$. A larger $\mathcal{U}_{y_{1}^{(l)}}(\varphi_{l})$ and a smaller $\sum_{j=l}^{l+l^{\prime}}{D_{j}^{*}}$ can lead to a larger SLA price, and vice versa; this is one basis to analyze the concrete results in Table~\ref{table-PC1}. On the other hand, the ratios of the maximum possible utilizations of a server under the 1st, 2nd, 3rd, 4th, 5th, and 6th SLAs to the utilization under the 1st SLA are respectively 1, 2.846, 5.217, 8.944, 13.91, and 19.26. The number of servers assigned for each SLA, the SLA price, and the server utilization together determine the revenue of a CSP. In contrast, the standard on-demand model is a special case and corresponds to the first SLA of our model; the related simulation results to be used are listed in Table~\ref{table-PC2}. For example, the first column says that, under the compact case where $m=1000$ and $\beta=0.55$, the optimal solution is specified by the tuple $(2, 39.81, 31.95, 1000)$. This tuple implies that, partial jobs of the first two users will be served with a price of 39.81, and, the CSP accepts jobs at a rate of 39.81 with all 1000 servers allocated.

\vspace{0.1em}\noindent\textbf{Phenomenon (\rmnum{1}).} Now, we begin to observe the results in the loose case with $\beta=0.55$, i.e., the 2-3 rows of Table~\ref{table-PC1}. With more servers available to process jobs (i.e., a larger $m$), a CSP is allowed to admit more users (see the purpled values under the 6th SLA).

These accepted users are distributed to different SLAs, and, under the same SLA $l$, we could see a larger $y_{1}^{(l)}$ in the case with a larger $m$. To a large extent, this leads to a smaller price for the case with a larger $m$. For example, the price of the 4th SLA in the case with $m=1000$ is $\mathcal{U}_{10}(\varphi_{4}) - \left(\mathcal{U}_{10}(\varphi_{5})-\mathcal{U}_{13}(\varphi_{5})\right) - \left(\mathcal{U}_{13}(\varphi_{6})-\mathcal{U}_{15}(\varphi_{6})\right) = 19.2866 - 2.1428 - 0.9883 = 16.1555$; its price in the case with $m=2000$ is $\mathcal{U}_{14}(\varphi_{4}) - \left(\mathcal{U}_{14}(\varphi_{5})-\mathcal{U}_{17}(\varphi_{5})\right) - \left(\mathcal{U}_{17}(\varphi_{6})-\mathcal{U}_{18}(\varphi_{6})\right) = 15.8671 - 1.7629 - 0.4165 = 13.6878$. QoS-differentiated pricing allows admitting more users than on-demand pricing where only the first SLA is offered; when $m$ increases from 1000 to 2000, the difference of the two values $y_{1}^{(l)}$ ($l=$ 3, 4, 5, 6) in the former case (illustrated in Table~\ref{table-PC1}) is larger than the difference of the two values $y_{1}^{(1)}$ in the latter case (illustrated in Table~\ref{table-PC2}).

Since $y_{1}^{(l)}$ can greatly affect the price in our simulation, the degree of price deduction in the former case is larger the one in the latter case. This is also shown by a larger unit revenue decline; in particular, using the results in Table~\ref{table-PC2}, the unit revenue decreases by 9.31\% when the number $m$ of possessed servers increases from 1000 to 2000; with Table~\ref{table-PC1}, the unit revenues respectively under the 3rd, 4th, 5th, and 6th SLAs decrease by 19.75\%, 15.28\%, 15.22\%, and 14.90\%. Furthermore, we could conclude from the rows 2-3 of Table~\ref{table-PC1} that, the unit revenue is increasing from a lower SLA (a higher QoS requirement) to a higher SLA; in particular, the unit revenues in the 2nd ({\em resp.} 3rd) row are respectively 3.4683, 4.6171, 5.6044, and 5.9323 ({\em resp.} 1.8182, 2.7833, 3.9114, 4.7515, and 5.0485). Overall, it happens that the revenue improvement decreases as $m$ becomes larger, after these unit revenues are weighted by the numbers of assigned servers.


\vspace{0.1em}\noindent\textbf{Phenomenon (\rmnum{2}).} Let us compare the results in the rows 3 and 4 where users respectively have medium and high sensitivities to latency ($\beta$ equals 0.55 and 0.2 respectively).

In the case that users are more sensitive to latency (i.e., $\beta=0.2$), (\rmnum{1}) only three SLAs (i.e., the 2nd, 3rd, and 4th SLAs) are provided to users (the $5$-th and 6-th SLAs with the largest waiting times will not be offered) and (\rmnum{2}) less users are accepted to rejected some users with lower utilities (the first 15 users with the larger utilities are accepted), in contrast to the optimal solution in the case with $\beta=0.55$. These two points ensure that a smaller value $y_{1}^{(l)}$ is associated under each offered SLA, and the offered SLAs are also associated with smaller waiting times. As illustrated in Figure~\ref{Fig.17}, as $\beta$ becomes smaller, the utilities (or willingness to pay) of users decrease more dramatically with the increment of waiting time. So, these two points lead to that the prices and the unit revenues of the offered SLAs are still relatively high in the case with $\beta=0.2$.

However, as shown in the simulations, the revenue improvement still decreases as users become more sensitive to latency. In on-demand pricing, using Table~\ref{table-PC2}, we derive that the unit revenue decreases by the unit revenue decreases by 17.36\% when $\beta$ decreases from 0.55 to 0.2; with Table~\ref{table-PC1}, the unit revenues respectively under the 2nd, 3rd, and 4th SLAs decrease by 12.70\%, 20.66\%, and 29.66\%. Furthermore, we could conclude from the rows 3-4 of Table~\ref{table-PC1} that, the unit revenue is increasing from a lower SLA (a higher QoS requirement) to a higher SLA; in particular, the unit revenues in the 4nd ({\em resp.} 3rd) row are respectively 1.5873, 2.2083, and 2.7514 ({\em resp.} 1.8182, 2.7833, 3.9114, 4.7515, and 5.0485). Overall, it happens that the revenue improvement decreases as $\beta$ becomes smaller, after these unit revenues are weighted by the numbers of assigned servers.

\vspace{0.1em}\noindent\textbf{Phenomenon (\rmnum{3}).} We look at the rows 3 and 5 of Table~\ref{table-PC1} where users respectively have loose and compact weight distributions (also see the columns 2 and 4 of Table~\ref{table-PC2}); 18 users are accepted in the loose case while 10 more users are accepted in the compact case.

In the loose case, the ratios of the prices of these offered SLAs (the 3rd, 4th, 5th, and 6th SLAs) in QoS-differentiated pricing to the on-demand price in the column 2 of Table~\ref{table-PC2} are 46.25\%, 37.91\%, 29.83\%, and 23.07\%; in the compact case, these ratios are 50.43\%, 39.37\%, 30.66\%, and 23.53\%. In the compact case, more accepted users don't lead to that the price of each SLA reduce more heavily. In the loose case with $m=2000$ and $\beta=0.55$, the ratios of the unit revenues from the 2nd, 3rd, 4th, 5th, and 6th SLAs to the unit revenue where only the first SLA is offered are 1.5762, 2.4129, 3.3908, 4.1191, and 4.3765; in the compact case, these ratios of the 3rd, 4th, 5th, and 6th SLAs are 2.6313, 3.5213, 4.2498, and 4.5093. After these unit revenues are weighted by the number of servers assigned to each SLA, we can see that it happens that the revenue improvement increases as the weight distribution becomes more compact.

By the phenomenon here, we noticed another phenomenon that the 2nd SLA is offered in the loose case but it is not so in the compact case. To understand this phenomenon, we consider the case where the 2nd SLA is not offered but the 3rd, 4th, 5th, and 6th SLAs are offered in the loose case. Roughly, in this case, more users could be accepted since servers have higher utilizations under the 3rd, 4th, 5th, and 6th SLAs than the 2nd SLA; similar to our analysis of the first phenomenon, this could reduce the SLA's prices and the unit revenue on the whole.



\section{Concluding Remarks}
\label{sec.conclusion}


In cloud computing, there exist both latency-critical jobs and jobs that could tolerate some degree of delay. The resource efficiency of a system is much dependent on the job's latency requirement. From a user's point of view, posted pricing has high usability. In such context, we propose the first analytical model for QoS-differentiated posted pricing. We consider a cloud computing system where the system is partitioned into multiple independent subsystems.

Optimal schemes have been proposed to address properly two key, intertwined aspects of the model: (\rmnum{1}) the pricing of different levels of QoS requirements, and, (\rmnum{2}) the arrival rate of jobs accepted to be processed, in connection with the number of servers assigned to each QoS level. Queuing models let us derive a  general analytical framework, which adapts to several popular dispatching policies. Numerical simulations show that the revenue of a CSP could be improved by up to a five-fold increase, with an improved system's utilization. The analytical results of this paper provides a framework to easily evaluate the performance of QoS-differentiated pricing in cloud computing, given the computing capacity of a CSP and the environments that a CSP faces (the size of a users' population, the amount of demanded computing resource, and the  users' sensitivity to latency).

\appendices
\section{}

\begin{proof}[Proof of Lemma~\ref{lemma-optimal-structure-1}]
Recall that, our objective is maximizing (\ref{equa-max}) where $\boldsymbol{\theta}$ is pre-assigned; in terms of the decision variable $\Lambda_{l}^{+}$, the constraints that matter are \eqref{equa-constraint-1} and \eqref{constraint-3} alone. Hence, the optimal value of $\Lambda_{l}^{+}$ under an arbitrary $\boldsymbol{m}$ is the maximum possible value that simultaneously satisfies \eqref{equa-constraint-1} and \eqref{constraint-3}, i.e., it equals $\min\{ \Lambda_{l} ,\, m_{l}\cdot \lambda_{max}^{(l)} \}$. Here, in the first case that $m_{l}\cdot\lambda_{max}^{(l)}\leq \Lambda_{l}$, we have $\Lambda_{l}^{+}=m_{l}\cdot\lambda_{max}^{(l)}$ where $m_{l}\leq \overline{m}_{l}$.

In the second case that $m_{l}\cdot\lambda_{max}^{(l)}>\Lambda_{l}$, we have $\Lambda_{l}^{+}$ $=$ $\Lambda_{l}$. In the latter case, we also have $m_{l} = \overline{m}_{l}+1$ where $\lambda_{l}^{\prime}>0$; now, we begin to prove this by contradiction. We first prove $m_{l}=\lceil \Lambda_{l}/\lambda_{max}^{(l)} \rceil$, which is the minimum number of servers needed by Lemma~\ref{lemma-max-rate}. Otherwise, $m_{l} > \lceil \Lambda_{l}/\lambda_{max}^{(l)} \rceil$ and, there exists some $l^{\prime}$ ($\neq l$) such that $m_{l^{\prime}}\cdot \lambda_{max}^{(l)}<\Lambda_{l}$ by the assumption (\ref{equa-case}). By increasing $m_{l^{\prime}}$ by $m_{l} - \lceil \Lambda_{l}/\lambda_{max}^{(l)} \rceil$ and reducing $m_{l}$ to $\lceil \Lambda_{l}/\lambda_{max}^{(l)} \rceil$, this allows setting $\Lambda_{l^{\prime}}^{+}$ to a higher value $\min\{ \Lambda_{l^{\prime}} ,\, m_{l^{\prime}}\cdot \lambda_{max}^{(l^{\prime})} \}$ while the value of $\Lambda_{l}^{+}$ does not decrease, i.e., $\Lambda_{l}^{+} = \min\{ \Lambda_{l} ,\, m_{l}\cdot \lambda_{max}^{(l)} \} = \Lambda_{l}$; as a result, the value of the item $\theta_{l^{\prime}}\cdot\Lambda_{l^{\prime}}^{+}$ in the objective function (\ref{equa-max}) is increased, as well as the value of the objective function. This contradicts that $\boldsymbol{m}$ and $\boldsymbol{\Lambda^{+}}$ are an optimal solution and we thus have $m_{l}=\lceil \Lambda_{l}/\lambda_{max}^{(l)} \rceil$. Now, we prove $\lambda_{l}^{\prime}>0$ by contradiction. Otherwise, $\lambda_{l}^{\prime} = \Lambda_{l} - \overline{m}_{l}\cdot \lambda_{max}^{(l)} =0$ and we have $m_{l}=\overline{m}_{l}$, which contradicts the second case where $m_{l}\cdot\lambda_{max}^{(l)}>\Lambda_{l}$.
\end{proof}

\begin{proof}[Proof of Observation~\ref{lemma-reduction-AA}]
Since the mapping $f$ is a surjection, given an arbitrary feasible solution $x\in\mathcal{X}$ to the problem $A$, there exists a feasible solution $y$ to the problem $B$ where $x=f(y)$; here, the objective functions in the two problems $A$ and $B$ achieve the same value. Hence, the optimal value of the problem $B$ is an upper bound of the optimal value of $A$. On the other hand, since $f$ is a mapping from $\mathcal{Y}$ to $\mathcal{X}$, given an optimal solution $y^{*}\in\mathcal{Y}$ to $B$, $x=f(y^{*})$ is a feasible solution to $A$ where $f(y)\in\mathcal{X}$; here, the two objective functions in $A$ and $B$ still achieve the same value and hence $x=f(y^{*})$ is an optimal solution to the problem $A$ where $x^{*}=f(y^{*})$.
\end{proof}

\begin{proof}[Proof of Observation~\ref{transformation-1}]
The determinant of the matrix $A$ is $>0$ and it is invertible. Hence, one $\begin{bmatrix}
   x_{l}  &   y_{l} \\
\end{bmatrix}$ corresponds to one unique $\begin{bmatrix}
   m_{l}  &   \Lambda_{l}^{+}\\
\end{bmatrix}$, and vice versa. In the following, we first show that, given an arbitrary $\begin{bmatrix}
   x_{l}  &   y_{l} \\
\end{bmatrix}$ satisfying (\ref{constraint-4}) and (\ref{constraint-2-1}), the corresponding $\begin{bmatrix}
   m_{l}  &   \Lambda_{l}^{+}\\
\end{bmatrix}$ (that is defined by (\ref{constraint-2})) satisfies Lemma~\ref{lemma-optimal-structure-1}, as well as the constraints \eqref{equa-num-machine}, \eqref{equa-constraint-1}, and (\ref{constraint-3}). This shows that (\ref{constraint-2}) defines a mapping.

By (\ref{constraint-2}), we have $m_{l}=x_{l}+y_{l}$ and $\Lambda_{l}^{+}=\lambda_{max}^{(l)}\cdot x_{l} + \lambda_{l}^{\prime}\cdot y_{l}$. With (\ref{constraint-4}), we have $0\leq x_{l}\leq \overline{m}_{l}$ and $y_{l}\in\{0, 1\}$; further, the $m_{l}$ and $\Lambda_{l}^{+}$ here satisfy \eqref{equa-constraint-1}, and (\ref{constraint-3}) respectively by (\ref{floor-1}) and (\ref{floor-2}). Due to (\ref{constraint-2-1}), the $m_{1}$, $\cdots$, $m_{L}$ also satisfy \eqref{equa-num-machine}. By (\ref{constraint-4}), there are three possible cases for the values of $x_{l}$ and $y_{l}$: (\rmnum{1}) $x_{l}<\overline{m}_{l}$, $y_{l}=0$, (\rmnum{2}) $x_{l}=\overline{m}_{l}$ and $y_{l}=0$, and (\rmnum{3}) $x_{l}=\overline{m}_{l}$ and $y_{l}=1$. In the second case, $\lambda_{l}^{\prime}$ may be either $>0$ or $=0$; in the third case, $\lambda_{l}^{\prime}>0$. In the first case, we have $m_{l}=x_{l}<\overline{m}_{l}$ and $\Lambda_{l}^{+}=\lambda_{max}^{(l)}\cdot m_{l}$. In the second case, we have $m_{l}=x_{l}=\overline{m}_{l}$ and $\Lambda_{l}^{+}=\lambda_{max}^{(l)}\cdot m_{l}$. In the third case where $\lambda_{l}^{\prime}>0$, we have $m_{l}=x_{l}+y_{l}=\overline{m}_{l}+1$ and $\Lambda_{l}^{+}=\lambda_{max}^{(l)}\cdot \overline{m}_{l} + \lambda_{l}^{\prime}=\Lambda_{l}$ by the definition of $\overline{m}_{l}$. After checking the corresponding $m_{l}$ and $\Lambda_{l}^{+}$ in each of the above cases, we conclude that Lemma~\ref{lemma-optimal-structure-1} could be satisfied.

Secondly, we show that, given an arbitrary $\begin{bmatrix}
   m_{l}  &   \Lambda_{l}^{+}\\
\end{bmatrix}$ satisfying Lemma~\ref{lemma-optimal-structure-1}, \eqref{equa-num-machine}, \eqref{equa-constraint-1}, and (\ref{constraint-3}), the corresponding $\begin{bmatrix}
   x_{l}  &   y_{l} \\
\end{bmatrix}$ (defined by (\ref{constraint-2})) satisfies (\ref{constraint-4}) and (\ref{constraint-2-1}). This finally shows that (\ref{constraint-2}) defines a surjection. Since $m_{l}=x_{l}+y_{l}$ and due to \eqref{equa-num-machine}, the $x_{l}$ and $y_{l}$ here satisfy (\ref{constraint-2-1}). By Lemma~\ref{lemma-optimal-structure-1}, there are only two possible cases for the variable $\Lambda_{l}^{+}$: (\rmnum{1}) $\Lambda_{l}^{+} < \Lambda_{l}$, and (\rmnum{2}) $\Lambda_{l}^{+} = \Lambda_{l}$. In the first case, we have $\Lambda_{l}^{+}=m_{l}\cdot \lambda_{max}^{(l)}$ and $m_{l}\leq \overline{m}_{l}$ by Lemma~\ref{lemma-optimal-structure-1} and further have $x_{l}=m_{l}\leq \overline{m}_{l}$ and $y_{l}=0$, which satisfy (\ref{constraint-4}). In the second case, if $\Lambda_{l}=\overline{m}_{l}\cdot \lambda_{max}^{(l)}$, we have $\lambda_{l}^{\prime}=0$ and $m_{l}=\overline{m}_{l}$, and further have $x_{l}=\overline{m}_{l}$ and $y_{l}=0$ by (\ref{constraint-2}); otherwise, $\Lambda_{l} > \overline{m}_{l}\cdot \lambda_{max}^{(l)}$, we have $\lambda_{l}^{\prime}>0$ and $m_{l}=\overline{m}_{l}+1$ by Lemma~\ref{lemma-optimal-structure-1} and further have $x_{l}=\overline{m}_{l}$ and $y_{l}=1$. Here, (\ref{constraint-4}) is also satisfied.
 Finally, the lemma holds.
\end{proof}

\begin{proof}[Proof of Observation~\ref{transformation-3}]
We first prove that (\ref{equation-3}) defines a mapping. For an arbitrary $( \{ x_{l} \}_{l\in\mathcal{A}},$  $\{ y_{l} \}_{l\in\mathcal{B}} )$ that satisfies (\ref{constraint-4-2}) and~(\ref{constraint-2-2}), we will show that the corresponding $\left( \{ x_{l} \}_{l=1}^{L}, \{ y_{l} \}_{l=1}^{L} \right)$ (defined by (\ref{equation-3})) satisfies \eqref{constraint-4} and \eqref{constraint-2-1}. For all $l\in\{1, 2, \cdots, L\} - \mathcal{A}$ where $\overline{m}_{l}=0$, we have that $x_{l}=0$ satisfies constraint (\ref{constraint-4}); for all $l\in\{1, 2, \cdots, L\} - \mathcal{B}$ where $\lambda_{l}^{\prime}=0$, we have that $y_{l}=0$ also satisfies (\ref{constraint-4}). For all $x_{l}$ where $l\in\mathcal{A}$, and $y_{l}$ where $l\in\mathcal{B}$ and $\lambda_{l}^{\prime}>0$, they satisfy constraint (\ref{constraint-4-2}), also satisfying (\ref{constraint-4}). Hence, (\ref{equation-3}) defines a mapping.

Secondly, we prove that (\ref{equation-3}) defines a surjection. For an arbitrary $\left( \{ x_{l} \}_{l=1}^{L}, \{ y_{l} \}_{l=1}^{L} \right)$ that satisfies \eqref{constraint-4} and \eqref{constraint-2-1}, we will show that there exists a corresponding $( \{ x_{l} \}_{l\in\mathcal{A}},\,$  $\{ y_{l} \}_{l\in\mathcal{B}} )$ that satisfies (\ref{constraint-4-2}) and~(\ref{constraint-2-2}). For all $x_{l}$ where $l\in\mathcal{A}$, and $y_{l}$ where $l\in\mathcal{B}$, they satisfy constraint (\ref{constraint-4}), also satisfying (\ref{constraint-4-2}). In the case that $l\in \{1, 2, \cdots, L\} - \mathcal{A}$ where $\overline{m}_{l}=0$, we derive $x_{l}=0$ due to the constraint (\ref{constraint-4}) alone; in the case that $l\in\{1, 2, \cdots, L\} - \mathcal{B}$ where $\lambda_{l}^{\prime}=0$, we also derive $y_{l}=0$ due to (\ref{constraint-4}) alone. As a result, if $\left( \{ x_{l} \}_{l=1}^{L}, \{ y_{l} \}_{l=1}^{L} \right)$ satisfies (\ref{constraint-2-1}), $( \{ x_{l} \}_{l\in\mathcal{A}},$  $\{ y_{l} \}_{l\in\mathcal{B}} )$ also satisfies (\ref{constraint-2-2}). Finally, the lemma holds.
\end{proof}

\begin{proof}[Proof of Observation~\ref{transformation-2}]
We first prove that (\ref{relationship-3}) defines a mapping, i.e., given an arbitrary $(x_{l,1}, \cdots, x_{l, \overline{m}_{l}}, y_{l})$ that satisfies \eqref{constraint-bit}, \eqref{constraint-3-1} and \eqref{constraint-3-2}, the corresponding $x_{l}$ and $y_{l}$ satisfy \eqref{constraint-4-2} and \eqref{constraint-2-2}. Due to \eqref{constraint-bit}, we have $0\leq x_{l} = \sum_{j=1}^{\overline{m}_{l}}{x_{l,j}}\leq \overline{m}_{l}$ where $\overline{m}_{l}>0$ for all $l\in \mathcal{A}$; together with (\ref{constraint-3-2}), the constraint (\ref{constraint-4-2}) holds. Due to (\ref{constraint-3-1}), \eqref{constraint-2-2} also holds. Further, we prove that (\ref{relationship-3}) defines a surjection. For an arbitrary $(x_{l},\, y_{l})$ that satisfies \eqref{constraint-4-2} and \eqref{constraint-2-2} where $l\in\mathcal{A}$, it holds naturally that there exist feasible $x_{l,1}, \cdots, x_{l,\overline{m}_{l}}\in \{0, 1\}$ such that $x_{l} = \sum_{j=1}^{\overline{m}_{l}}{x_{l,j}}$, where $x_{l,1}, \cdots, x_{l,\overline{m}_{l}}$ and $y_{l}$ satisfy \eqref{constraint-bit}, (\ref{constraint-3-2}), and (\ref{constraint-3-1}).
\end{proof}

\begin{proof}[Proof of Theorem~\ref{theorem-capacity-planning-1}]
We first show that Algorithm~\ref{algo-prob33-solution} produces a feasible solution to Problem~\ref{prob33}. Algorithm~\ref{algo-prob33-solution} obeys the constraints \eqref{constraint-bit} and \eqref{constraint-3-1}, without accounting for \eqref{constraint-3-2}. However, we will show that \eqref{constraint-3-2} is still satisfied since $\zeta_{l,j}^{(1)}> \zeta_{l}^{(2)}$: this is due to that, if $\sum_{j=1}^{\overline{m}_{l}}{x_{l,j}}<\overline{m}_{l}$, only a part of variables $x_{l,1}, \cdots, x_{l,\overline{m}_{l}}$ are set to 1 and all the other variables are set to 0; then, the variable $y_{l}$ is not chosen and is set to 0.

Next, we prove the optimality by contradiction. Consider an optimal solution in which $\mathcal{D}^{\prime}$ denotes the variables whose values are set to 1 and the values of the other variables are set to 0. Here, $\mathcal{D}^{\prime}$ contains a variable whose coefficient is not among the largest $m$ coefficients. However, in this case, the objective function could achieve a higher value by set a variable with a larger coefficient to one and setting the value of this variable to 0, which contradicts that $\mathcal{D}^{\prime}$ defines an optimal solution.
\end{proof}

\begin{proof}[Proof of Theorem~\ref{theorem-capacity-planning-1}]
Algorithm~\ref{algo-prob-solution} first calls Algorithm~\ref{algo-prob33-solution} ({\em line 1}), which gives an optimal solution to Problem~\ref{prob33}, by Theorem~\ref{algo-prob33-solution}. Built on this, the lines 2-3 give an optimal solution to Problem~\ref{prob44} by Lemma~\ref{lemma-relationship-3}. By Lemma~\ref{lemma-relationship-4}, the lines 4-5 give an optimal solution to Problem~\ref{prob22}. Finally, by Lemma~\ref{lemma-relationship-2}, the lines 6-7 give an optimal solution to Problem~\ref{prob1}.
\end{proof}

\ifCLASSOPTIONcaptionsoff
  \newpage
\fi

\end{document}